\RequirePackage{amsmath}
\documentclass{elsarticle}

\usepackage{hyperref,amsmath,amssymb,breakurl,textcomp,centernot,graphicx,mathtools,cleveref,listings}

\usepackage[justification=centering]{caption}
\usepackage[cal=cm]{mathalfa}
\usepackage[usenames,dvipsnames]{color}

\usepackage{tikz}
\usetikzlibrary{shapes,arrows}
\usetikzlibrary{shapes.multipart}
\usetikzlibrary{positioning}
\usetikzlibrary{matrix,arrows,automata} 
\usetikzlibrary{arrows,decorations.pathmorphing,backgrounds,positioning,calc}

\usepackage{theorem}
\theorembodyfont{\rm}
\newtheorem{definition}{Definition}[section]
\newtheorem{lemma}[definition]{Lemma}

\newtheorem{theorem}[definition]{Theorem}
\newtheorem{example}[definition]{Example}

\newenvironment{proof}{\noindent\textbf{Proof}}{\hfill$\square$}

\newdimen\proofrulebreadth \proofrulebreadth=.05em
\newdimen\proofdotseparation \proofdotseparation=1.25ex
\newdimen\proofrulebaseline \proofrulebaseline=2ex
\newcount\proofdotnumber \proofdotnumber=3
\let\then\relax
\def\hfi{\hskip0pt plus.0001fil}
\mathchardef\squigto="3A3B
\newif\ifinsideprooftree\insideprooftreefalse
\newif\ifonleftofproofrule\onleftofproofrulefalse
\newif\ifproofdots\proofdotsfalse
\newif\ifdoubleproof\doubleprooffalse
\let\wereinproofbit\relax
\newdimen\shortenproofleft
\newdimen\shortenproofright
\newdimen\proofbelowshift
\newbox\proofabove
\newbox\proofbelow
\newbox\proofrulename
\def\shiftproofbelow{\let\next\relax\afterassignment\setshiftproofbelow\dimen0 }
\def\shiftproofbelowneg{\def\next{\multiply\dimen0 by-1 }%
\afterassignment\setshiftproofbelow\dimen0 }
\def\setshiftproofbelow{\next\proofbelowshift=\dimen0 }
\def\setproofrulebreadth{\proofrulebreadth}

\def\prooftree{
\ifnum  \lastpenalty=1
\then   \unpenalty
\else   \onleftofproofrulefalse
\fi
\ifonleftofproofrule
\else   \ifinsideprooftree
        \then   \hskip.5em plus1fil
        \fi
\fi
\bgroup
\setbox\proofbelow=\hbox{}\setbox\proofrulename=\hbox{}%
\let\justifies\proofover\let\leadsto\proofoverdots\let\Justifies\proofoverdbl
\let\using\proofusing\let\[\prooftree
\ifinsideprooftree\let\]\endprooftree\fi
\proofdotsfalse\doubleprooffalse
\let\thickness\setproofrulebreadth
\let\shiftright\shiftproofbelow \let\shift\shiftproofbelow
\let\shiftleft\shiftproofbelowneg
\let\ifwasinsideprooftree\ifinsideprooftree
\insideprooftreetrue
\setbox\proofabove=\hbox\bgroup$\displaystyle 
\let\wereinproofbit\prooftree
\shortenproofleft=0pt \shortenproofright=0pt \proofbelowshift=0pt
\onleftofproofruletrue\penalty1
}

\def\eproofbit{
\ifx    \wereinproofbit\prooftree
\then   \ifcase \lastpenalty
        \then   \shortenproofright=0pt  
        \or     \unpenalty\hfil         
        \or     \unpenalty\unskip       
        \else   \shortenproofright=0pt  
        \fi
\fi
\global\dimen0=\shortenproofleft
\global\dimen1=\shortenproofright
\global\dimen2=\proofrulebreadth
\global\dimen3=\proofbelowshift
\global\dimen4=\proofdotseparation
\global\count255=\proofdotnumber
$\egroup  
\shortenproofleft=\dimen0
\shortenproofright=\dimen1
\proofrulebreadth=\dimen2
\proofbelowshift=\dimen3
\proofdotseparation=\dimen4
\proofdotnumber=\count255
}

\def\proofover{
\eproofbit 
\setbox\proofbelow=\hbox\bgroup 
\let\wereinproofbit\proofover
$\displaystyle
}%
\def\proofoverdbl{
\eproofbit 
\doubleprooftrue
\setbox\proofbelow=\hbox\bgroup 
\let\wereinproofbit\proofoverdbl
$\displaystyle
}%
\def\proofoverdots{
\eproofbit 
\proofdotstrue
\setbox\proofbelow=\hbox\bgroup 
\let\wereinproofbit\proofoverdots
$\displaystyle
}%
\def\proofusing{
\eproofbit 
\setbox\proofrulename=\hbox\bgroup 
\let\wereinproofbit\proofusing
\kern0.3em$
}

\def\endprooftree{
\eproofbit 
  \dimen5 =0pt
\dimen0=\wd\proofabove \advance\dimen0-\shortenproofleft
\advance\dimen0-\shortenproofright
\dimen1=.5\dimen0 \advance\dimen1-.5\wd\proofbelow
\dimen4=\dimen1
\advance\dimen1\proofbelowshift \advance\dimen4-\proofbelowshift
\ifdim  \dimen1<0pt
\then   \advance\shortenproofleft\dimen1
        \advance\dimen0-\dimen1
        \dimen1=0pt
        \ifdim  \shortenproofleft<0pt
        \then   \setbox\proofabove=\hbox{%
                        \kern-\shortenproofleft\unhbox\proofabove}%
                \shortenproofleft=0pt
        \fi
\fi
\ifdim  \dimen4<0pt
\then   \advance\shortenproofright\dimen4
        \advance\dimen0-\dimen4
        \dimen4=0pt
\fi
\ifdim  \shortenproofright<\wd\proofrulename
\then   \shortenproofright=\wd\proofrulename
\fi
\dimen2=\shortenproofleft \advance\dimen2 by\dimen1
\dimen3=\shortenproofright\advance\dimen3 by\dimen4
\ifproofdots
\then
        \dimen6=\shortenproofleft \advance\dimen6 .5\dimen0
        \setbox1=\vbox to\proofdotseparation{\vss\hbox{$\cdot$}\vss}%
        \setbox0=\hbox{%
                \advance\dimen6-.5\wd1
                \kern\dimen6
                $\vcenter to\proofdotnumber\proofdotseparation
                        {\leaders\box1\vfill}$%
                \unhbox\proofrulename}%
\else   \dimen6=\fontdimen22\the\textfont2 
        \dimen7=\dimen6
        \advance\dimen6by.5\proofrulebreadth
        \advance\dimen7by-.5\proofrulebreadth
        \setbox0=\hbox{%
                \kern\shortenproofleft
                \ifdoubleproof
                \then   \hbox to\dimen0{%
                        $\mathsurround0pt\mathord=\mkern-6mu%
                        \cleaders\hbox{$\mkern-2mu=\mkern-2mu$}\hfill
                        \mkern-6mu\mathord=$}%
                \else   \vrule height\dimen6 depth-\dimen7 width\dimen0
                \fi
                \unhbox\proofrulename}%
        \ht0=\dimen6 \dp0=-\dimen7
\fi
\let\doll\relax
\ifwasinsideprooftree
\then   \let\VBOX\vbox
\else   \ifmmode\else$\let\doll=$\fi
        \let\VBOX\vcenter
\fi
\VBOX   {\baselineskip\proofrulebaseline \lineskip.2ex
        \expandafter\lineskiplimit\ifproofdots0ex\else-0.6ex\fi
        \hbox   spread\dimen5   {\hfi\unhbox\proofabove\hfi}%
        \hbox{\box0}%
        \hbox   {\kern\dimen2 \box\proofbelow}}\doll%
\global\dimen2=\dimen2
\global\dimen3=\dimen3
\egroup 
\ifonleftofproofrule
\then   \shortenproofleft=\dimen2
\fi
\shortenproofright=\dimen3
\onleftofproofrulefalse
\ifinsideprooftree
\then   \hskip.5em plus 1fil \penalty2
\fi
}

\newcommand\kleene[1]{{#1}^{*}}
\newcommand\infern[3]{\begin{prooftree} #1 \justifies #2 \using(#3) \end{prooftree}}
\newcommand\infer[2]{\begin{prooftree} #1 \justifies #2  \end{prooftree}}

\newcommand\ptos[1]{\overline{#1}} 

\newcommand\card[1]{\left\vert{#1}\right\vert}
\newcommand\length[1]{\left\vert{#1}\right\vert}

\newcommand\loopnode{p_{\ell}}
\newcommand\nondetnode{p_{N}}

\newcommand\nfatrans{\twoheadrightarrow}
\newcommand\nfamor[3]{#1\mkern2mu:\mkern2mu#2\mkern2mu \nfatrans \mkern2mu #3}
\newcommand\dfamor[3]{#1\mkern2mu:\mkern2mu#2\mkern2mu 	\dfaarrow \mkern2mu #3}
\newcommand\twopaths[3]{#1\mkern2mu:\mkern2mu#2\mkern2mu 	\rightrightarrows \mkern2mu #3}
\newcommand\omor[3]{#1\mkern2mu:\mkern2mu#2\mkern2mu 	\mapsto \mkern2mu #3}

\newcommand\Acc{\mathbf{Acc}} 
\newcommand\nondeta{p_{N1}} 
\newcommand\nondetb{p_{N2}} 
\newcommand\startnode{p_{0}}

\newcommand\nfast{NFA node}

\newcommand\machname{backtracking machine}

\newcommand\mstat[1]{\Phi_{#1}} 
\newcommand\pseq[1]{\beta_{#1}} 
\newcommand\piseqs[1]{\sigma_{#1}} 

\newcommand\empseq{\varepsilon}
\newcommand\reachpre{\mathcal R} 
\newcommand\elems[1]{\mathsf{Set}(#1)}
\newcommand\ordtrans{\mathrel{\Rsh}} 
\newcommand\pump[1]{y_{#1}}
\newcommand\failstate{\mstat{\mathrm{fail}}}

\newcommand\dfaarrow{\Rrightarrow}
\newcommand\cstep[1]{\stackrel{#1}{\rightsquigarrow}} 

\newcommand\treemor[3]{#1 : #2 \mathrel{\triangle} #3} 

\newcommand\prestatone{\mstat{x}} 
\newcommand\pumpstatone{\mstat{y1}} 
\newcommand\pumpstattwo{\mstat{y1a}} 
\newcommand\pumpstatthree{\mstat{y2}} 
\newcommand\suffstat{\mstat{\mathrm{fail}}} 
\newcommand\twoarrow{\mathrel{\twoheadrightarrow}_{\mathcal{\mathbf 2}} }
\newcommand\threearrow{\mathrel{\twoheadrightarrow}_{\mathcal{\mathbf 3}} }
\newcommand\twomor[3]{#1 : #2 \twoarrow #3} 
\newcommand\threemor[3]{#1 : #2 \threearrow #3} 

\newcommand\glab[1]{\ptos #1}
\newcommand\gcod[1]{\mathrm{cod}(#1)}
\newcommand\gdom[1]{\mathrm{dom}(#1)}
\newcommand\gnodes[1]{\mathsf{nodes}(#1)}
\newcommand\gshared[1]{\mathcal{S}(#1)}
\newcommand\gfree[1]{\mathcal{F}(#1)}

\newcommand\empstring{\varepsilon} 
\newcommand\mklist[1]{[#1]} 
\newcommand\emplist{\mklist{\,}} 

\begin{document}

\title{Static Analysis for Regular Expression Exponential Runtime via Substructural Logics}

\author[bham]{Asiri Rathnayake}
\ead{apr015@cs.bham.ac.uk}

\author[bham]{Hayo Thielecke}
\ead{h.thielecke@cs.bham.ac.uk}

\address[bham]{University of Birmingham, Birmingham B15 2TT, United Kingdom}

\begin{abstract}
Regular expression matching using backtracking can have exponential runtime, leading to an algorithmic complexity attack known as REDoS in the systems security literature. In this paper, we present a static analysis that detects whether a given regular expression can have exponential runtime for some inputs. The analysis works by forming powers and products of transition relations and thereby reducing the REDoS problem to reachability. The correctness of the analysis is proved using a substructural calculus of search trees, where the branching of the tree causing exponential blowup is characterized as a form of non-linearity.
\end{abstract}

\maketitle

\section{Introduction}
Regular expressions are everywhere. Yet the backtracking virtual machines that are used to match them (in Java, .NET and other frameworks) are very different from the DFA (deterministic finite automata) construction used in compiling. Whereas DFAs run in linear time but may be expensive to construct, backtracking matchers have low initial cost, but may have exponential runtime for some inputs~\cite{2007_regex_cox}. This is a problem when such matchers may be exposed to malicious input, say over a network, as an attacker could craft an input in order for the matcher to take exponential time. This problem is known as REDoS, short for Regular Expression Denial-of-Service.

For a straightforward example of an exponential blowup, consider the following regular expression:
\[
\kleene{(a \mid b \mid ab)}c
\]
Matching this expression against input strings of the form $(ab)^n$ leads the Java virtual machine to a halt for very moderate values of $n$ ($\sim 50$) on a contemporary computer. Certain other backtracking matchers like the PCRE library and the matcher available in the .NET platform seem to handle this particular example well. However, the ad-hoc nature of the workarounds implemented in these frameworks are easily exposed with a slightly complicated expression / input combination:
\[
\kleene{(a \mid b \mid ab)}bc
\]
This expression, when matched against input strings of the form $(ab)^nac$, leads to exponential blowups on all the three matchers mentioned.

The REDoS analysis builds on the idea of nondeterministic Kleene expressions. When matching the input string $ab$ against the Kleene expression $\kleene{(a \mid b \mid ab)}$, a match could be found by taking either of the two different paths through the corresponding NFA (nondeterminstic finite automata). If we repeat this string to form $abab$, now there are four different paths through the NFA; this process quickly builds up to an exponential amount of paths through the NFA as the \textit{pumpable string} $ab$ is repeated.
\begin{definition}\label{def_pumpable}
For a given regular expression, a pumpable string is a substring $y$ of an input string of the form $xy^nz$. The pumpable part $y$ may be repeated as many times as needed ($n$) to form larger input strings while maintaining an exponential growth of search space in a backtracking matcher.
\end{definition}
A matcher based on DFAs would not face a difficulty in dealing with such expressions since the DFA construction eliminates such redundant paths. However, these expressions can be fatal for backtracking matchers based on NFAs, as their operation depends on performing a depth-first traversal of the entire search space.

We think of the various phases of the analysis as very simple and non-standard logics for judgements for different implications of the form:
\[
w : p_{1} \to p_{2}
\] 
Here $p_{1} \to p_{2}$ is a proposition and $w$ is a proof of it, which we will call its realizer. In this way, we can focus on what the analysis tries to construct, not how. Hence the analysis can be seen as a form of proof search, and it is implemented via straightforward closure algorithms.
  
A second use of logic or type theory in this paper comes in when proving the soundness of the analysis, when we need to show that the constructed string really leads to exponential runtime. While the \machname{} that we use as an idealization of backtracking matchers (like those in the Java platform) is not very complicated, it is not straightforward to reason about how it behaves on some constructed malicious input. This is because  the machine traverses the search tree in a depth-first strategy, whereas the attack string is best understood in terms of a  composition of horizontal slices of the search tree. To reason compositionally, we first introduce a calculus of search trees, inspired by substructural logics. In a nutshell, the existence of a pumpable string as part of a REDoS vulnerability amounts to the existence of a non-linear derivation in the search tree logic, essentially as in a derivation of this form:
\[
\infer{p}{p,p}
\]
Thus we can reason about the exponential growth of the search tree in a compositional, logical style, separate of the search strategy of the backtracking matcher. The exponential runtime of the machine then follows due to the fact that the runtime is at least the width of the search tree if the matcher is forced to explore the whole tree.

 
\subsection*{Outline of the paper}
\Cref{background} presents some required background on regular expression matching in a form that will be convenient for our purpose. We then define the three phases (prefix, pumping, and suffix construction) of our REDoS analysis in \Cref{analysis} and validate it on some examples in \Cref{examples}. \Cref{soundness} and \Cref{completeness} prove the soundness and the completeness of the analysis using a substructural calculus of search trees. \Cref{implementation} presents a brief overview of the OCaml implementation of the analysis and
the practical performance of our tool. We conclude with a discussion of related work in \Cref{related} and directions for further work in \Cref{further}.

\begin{figure}[tp]
\begin{center}
\begin{tabular}{ll}
\hline
\hline
$a$, $b$, $c$& input symbols\\
$w$, $x$, $y$, $z$\;\;& strings of input symbols\\
$p_{},q$& NFA nodes or states\\
$\pseq{}$, $\theta$& ordered sequences $p_{1}\ldots p_{n}$ of NFA states\\
$\mstat{}$& sets $\{p_{1},\ldots, p_{n}\}$ of NFA states\\
$\piseqs{}$& sequences of state/index pairs $(p,j)$\\
$\varepsilon$& empty word or sequence\\
$\nfatrans$& NFA transition relation (Definition~\ref{nfarelation})\\
$\mapsto$& ordered NFA transition function (Definition~\ref{defnfa})\\
$\ordtrans$& ordered multistate transition function (Definition~\ref{orderedtrans})\\
$\dfaarrow$& multistate transition function (Definition~\ref{powerdfa})\\
$\rightsquigarrow$& transitions of \machname{} (Definition~\ref{defbacktrack})\\
\hline
\hline
\end{tabular}
\caption{Notational conventions}
\label{notation}
\end{center}
\end{figure}

\section{Basic constructs}\label{background}
This section presents some background material that will be needed for the analysis, such as nondeterministic automata. \Cref{notation} gives an overview of notation. We assume that the regular expression has been converted into an automaton following one of the standard constructions. 

\subsection{The ordered NFA} \label{nfa}
The usual text-book definitions of NFAs do not impose any ordering on the transition function. For an example, a traditional NFA for the regular expression $a(bc \mid bd)$ would not prioritize any of the two transitions available for character $b$ over the other. Since backtracking matchers follow a greedy left-to-right evaluation of alternations, the alternation operator effectively becomes non-commutative in their semantics for regular expressions. Capturing this aspect in the analysis requires a specialized definition of NFAs.

If we are only concerned about acceptance, Kleene star is idempotent and alternation is commutative. If we are interested in exponential runtime, they are not. The non-commutativity of alternation is not that surprising in terms of programming language semantics, as Boolean operators like \verb|&&| in C or \texttt{andalso} in ML have a similar semantics: first the left alternative is evaluated, and if that does not evaluate to true, the right alternative is evaluated.
Since in our tool the NFA is constructed from the syntax tree, the order is already available in the data structures. The children of a \nfast{} have a left-to-right ordering.

\begin{definition}[Ordered NFA]
\label{defnfa}
An ordered NFA $\mathcal N$ consists of a set of states, an initial state $\startnode$, a set of accepting states $\Acc$ and for each input symbol $a$ a transition function from states to sequences of states. We write this function as 
\[
a : p \mapsto q_{1}\ldots q_{n}
\]
For each input symbol $a$ and current NFA state $p$, we have a sequence of successor states $q_{i}$. The order is significant, as it determines the order of backtracking.
\end{definition}

In the textbook definition of an $\varepsilon$-free NFA, the NFA  has a transition function $\delta$ of type
\[
\delta : (Q\times\Sigma) \to 2^{Q}
\]
where $Q$ is the set of states and $\Sigma$ the set of input symbols. Here we have imposed an order on the sets in the image of the function, replacing $2^{Q}$ by $Q^{*}$, curried the function, and swapped the order of $Q$ and $\Sigma$.
\[
\begin{array}{rcccc}
\Sigma &\to&  (Q &\to& Q^{*})
\\
a &\mapsto& p & \mapsto & q_{1}\ldots q_{n}
\end{array}
\]

\begin{definition} \label{nfarelation}
The nondeterministic transition relation of the NFA is given by the following inference:
\[
\infer{
a : p \mapsto q_{1}\ldots q_{n}
}{
\nfamor a p{q_{i}}
}
\]
\end{definition}
Note however, that we cannot recover the ordering of the successor states $q_{i}$ from the nondeterministic transition relation. In this regard, the NFA on which the matcher is based has a little extra structure compared to the standard definition of NFA in automata theory.
If we know that
\[
\nfamor a{p}{q_{1}}
\mbox{ and }
\nfamor a{p}{q_{2}}
\]
we cannot decide whether the ordered transition is
\[
\omor a{p}{q_{1}\,q_{2}}
\mbox{ or }
\omor a{p}{q_{2}\,q_{1}}
\]

\begin{definition} \label{nfarelation2}
The transition relation $\nfatrans$ is extended to strings with the following inference rules:
\[
\infer{}{\nfamor{\empstring}{p}{p}} \qquad \infer{\nfamor{w}{p}{q} \quad \nfamor{a}{q}{q'}}{\nfamor{wa}{p}{q'}}
\]
\end{definition}

We define ordered multistates $\pseq{}$ as finite sequences of NFA states $p$, where the order is significant. In constrast, multistates $\mstat{}$ represent sets of NFA states, so they can be represented as lists, but are identified up to reordering. Each ordered multistate $\pseq{}$ can be turned into a multistate given by the set of its elements. We write this set as $\elems{\pseq{}}$. If 
\[
\pseq{} = p_{1}\ldots p_{n}
\]
then
\[
\elems{\pseq{}} = \{ p_{1},\ldots, p_{n} \}
\]
The difference between ${\pseq{}}$ and $\elems{\pseq{}}$ may appear small, but the notion of equality for sets is less fine-grained than for sequences, which has an impact on the search space that the analysis has to explore.

The transition relation $\ordtrans$ on ordered multistates is defined in the following manner. Let $\ggg$ be an operator on ordered multistates:
\begin{eqnarray*}
\emplist &\ggg& \emplist\\
p\,\pseq{} &\ggg& p\,\pseq{}' \qquad \mbox{if } (\not\exists \pseq{1}, \pseq{2} \,.\, \pseq{} = \pseq{1}\,p\,\pseq{2}) \land \pseq{} \ggg \pseq{}'\\
p\,\pseq{} &\ggg& \pseq{}' \qquad \mbox{if } (\exists \pseq{1}, \pseq{2} \,.\, \pseq{} = \pseq{1}\,p\,\pseq{2}) \land p\,\pseq{1}\,\pseq{2} \ggg \pseq{}'
\end{eqnarray*}
The operator $\ggg$ removes all but the leftmost occurrences in sequences. Note that termination is guaranteed since $\ggg$ is applied on shorter sequences on the R.H.S. Moreover, in each reduced sequence each $p$ can appear at most once, so there are only finitely many sequences that can be reached for any NFA.
\begin{definition}
\label{orderedtrans} 
The transition relation $\ordtrans$ for ordered multistates is defined as:
\[
\infer{\pseq{1} = (p_{1}\ldots p_{n}) \qquad a : p_{i} \mapsto \theta_{i} \qquad (\theta_{1}\ldots \theta_{n}) \ggg \pseq{2}}
{a : \pseq{1} \ordtrans \pseq{2}}
\]
The $\ordtrans$ transition is extended to strings with the rules:
\[
\infer{}{\varepsilon : \pseq{} \ordtrans \pseq{}}
\qquad
\infer{w: \pseq 1 \ordtrans \pseq 2 \qquad a: \pseq 2 \ordtrans \pseq 3}
{(w\,a): \pseq 1 \ordtrans \pseq 3}
\]
\end{definition}
The transition relation $\ordtrans$ can be thought of as an ordered variant of the traditional DFA transition relation encountered in powerset construction:
\begin{definition}[DFA Transition $\dfaarrow$] \label{powerdfa}
The transition relation $\dfaarrow$ is defined as
\[
w : \mstat 1 \dfaarrow \mstat 2
\]
if and only if 
\[
\mstat 2 = 
\{
p_{2} \mid \exists p_{1} \in \mstat 1.\nfamor w{p_{1}}{p_{2}}
\}
\]
\end{definition}
 
\subsection{The backtracking abstract machine}\label{secmathcer}
The analysis assumes exact matching semantics of regular expressions. Given regular expression $e$ and the input string $w$, the matcher is required to find a match of the entire string, as opposed to a sub-string. Most practical matchers search for a sub-match by default. However, such behavior can be modeled in exact matching semantics by augmenting the regular expression with match-all constructs at either end of the expression, as in $(\kleene{.}e\kleene{.})$. Practical implementations offer special ``anchoring" constructs that allow regular expression authors to enforce exact matching semantics. For an example, expressions of the form $(\text{\textasciicircum}e\text{\$})$ require them to be matched against the entire input string.

While the theoretical formulation of our analysis assumes exact matching semantics (thus avoiding unnecessary clutter), our implementation assumes sub-match semantics, since it is more useful in practice. The translation between the two semantics is quite straightforward.

\begin{definition}[Backtracking abstract machine]
\label{defbacktrack}
Given an ordered NFA, the \machname{} is defined  as follows. We assume an input string $w$ as given. Machine transitions may depend on $w$, but it does not change during transitions, so that we do not explicitly list it as part of the machine state. The input symbol at position $j$ in $w$ is written as $w[j]$ ($j$ is 0-based).
\begin{itemize}
\item States of the \machname{} are finite sequences of the form
\[
\piseqs{} = (p_{0}, j_{0}) \ldots (p_{n}, j_{n})
\]
where each of the $p_{i}$ is an NFA state and each of the $j_{i}$ is an index into the current input string. We refer to individual $(p, i)$ pairs as frames (as in stack frames).
\item The initial state of the machine is the sequence of length 1 containing the frame:
\[
(\startnode,0)
\]
\item The machine has matching transitions, which are inferred from the transition function of the ordered NFA as follows:
\[
\infer{
w[j]=a
\qquad
a: p \mapsto q_{1}\ldots q_{n}
}{
w \Vdash
(p,j)\,\piseqs{}
\rightsquigarrow
(q_{1},j+1)\ldots (q_{n},j+1)\,\piseqs{}
}
\]
\item The machine has failing transitions, of the form:
\[
(p,j)\;\piseqs{} 
\rightsquigarrow
\piseqs{} 
\]
where there is no transition of the form $w[j]: p \mapsto q_{1}\ldots q_{n}$ \textit{OR} $j$ is the length of $w$ and $p\notin\Acc$. \label{fix_bam_failtrans}
\item
Accepting states are of the form:
\[
w \Vdash
(p,j)\;\piseqs{}
\]
where
$p\in \Acc$
and
$j$ is the length of $w$.
\item Transition sequences in $n (n > 0)$ steps are written as $\cstep n$ and inferred using the following rules: \label{fix_bam_ntrans}
\[
\infer{w \Vdash \piseqs{} \cstep{} \piseqs{}'}{w \Vdash \piseqs{} \cstep{1} \piseqs{}'} 
\qquad
\infer{w \Vdash \piseqs{1} \cstep n \piseqs{2} \qquad w \Vdash \piseqs{2} \cstep m \piseqs{3}}
{w \Vdash \piseqs{1} \cstep{n + m} \piseqs{3}} 
\]
We write  $w \Vdash
\piseqs{1} \cstep *\piseqs{2}$  for $\exists n.w \Vdash
(\piseqs{1} \cstep n\piseqs{2})$. 
\item
Final states are either accepting or the empty sequence.
\end{itemize}
\end{definition}

The state of the \machname{} is a stack that implements failure continuations. When the state is of the form $(p,j)\piseqs{}$, the machine is currently trying to match the symbol at position $j$ in state $p$. Should this match fail, it will pop the stack and proceed with the failure continuation $\piseqs{}$.

\begin{lemma}
\label{lem_backtrack_extend}
For any backtracking machine run (i.e. a transition sequence):
\[
w \Vdash \piseqs{} \cstep{n} \piseqs{}'
\]
And for any $\bar{\piseqs{}}$, the following run also exists:
\[
w \Vdash \piseqs{}\,\bar{\piseqs{}} \cstep{n} \piseqs{}'\,\bar{\piseqs{}}
\]
\end{lemma}
\begin{proof}
Observe that each transition taken by the the first machine can be simulated on the extended machine. Moreover, each transition of the extended machine leaves the additional $\bar{\piseqs{}}$ untouched. 
\end{proof}

The backtracking machine definition leaves a lot of leeway to the implementation. Implementation details are abstracted in the ordered transition relation. The most important choice in the definition is that the machine performs a depth-first traversal of the search tree. In principle, a backtracking matcher could also use breadth-first search. In that case, our REDoS analysis would not be applicable, and such matchers may avoid exponential run-time. However, the space requirements of breadth-first search are arguably prohibitive. A more credible alternative to backtracking matchers is Thompson's matcher~\cite{1968_thompson,2007_regex_cox,2009_regex_cox}, which is immune to REDoS. However the relative inflexibility of the lockstep algorithm (when supporting extended, non-regular pattern matching constructs) has made it less popular among practical regular expression libraries. The REDoS problem in the backtracking paradigm therefore remains quite significant.

\section{The REDoS analysis} \label{analysis}
At a high-level, the analysis can be thought to work by exploring all the Kleene sub-expressions within an input expression and then analysing each of those sub-expressions for an exponential vulnerability. Given a regular expression input, for each Kleene sub-expression $\kleene{e_2}$, there exists a prefix expression $e_1$ and a suffix expression $e_3$; the prefix expression is what a backtracking matcher has to match in order to reach the Kleene expression, the suffix expression is what a backtracking matcher has to match in order to complete an overall match. For each such decomposition of the form $e_1\kleene{e_2}e_3$, the analysis attempts to derive an attack string of the form:
\[
xy^nz
\]
The presence of a pumpable string $y$ signals the analyser that a corresponding prefix $x$ and a suffix $z$ need to be derived in order to form the final attack string configuration. The requirements on the different segments of the attack string are as follows:
\begin{eqnarray*}
x &:& x \in L(e_1) \label{nssprefix}\\
y &:& y \in L(\kleene{e_2}) \quad \text{(with $b > 1$ paths)} \label{nsspumpable}\\
z &:& xy^nz \not\in L(e_1\kleene{e_2}e_3) \label{nsssuffix}
\end{eqnarray*}
Intuitively, the prefix $x$ leads a backtracking matcher to a point where it has to match the (vulnerable) Kleene expression $\kleene{e_2}$. At this point the matcher is presented with $n$ ($n > 0)$ copies of the pumpable string $y$, increasing the search space of the matcher to the order of $b^n$. At the end of each of the search attempts (paths through the NFA), the suffix $z$ causes the matcher to backtrack, forcing an exploration of the entire search space. 

\subsection{The phases of the REDoS analysis}
Overall, the REDoS analysis of a node $\loopnode$ (loop node) consists of three phases. The phases all work by incrementally exploring a transition relation. These relations are the DFA transition relation $\dfaarrow$ (Definition~\ref{powerdfa}) and its ordered variant $\ordtrans$ (Definition~\ref{orderedtrans}). The three analysis phases construct a REDoS prefix $x$, a pumpable string $\pump{}$ and a REDoS suffix $z$:
\begin{equation*}
\begin{array}{rcl}
\text{Prefix analysis} &\;& \left\{
\begin{array}{rcll}
x &:& \startnode \ordtrans (\pseq{}\;\loopnode\;\pseq{}')\\[1em]
\end{array} \right.\\
\text{Pumpable analysis} &\;& \left\{
\begin{array}{rcll}
\pump{1} &:& \prestatone \dfaarrow \pumpstatone
& \mbox{where }\prestatone =  \elems{\pseq{}\,\loopnode}\\[1ex]
{a} &:& \pumpstatone \dfaarrow \pumpstattwo\\[1ex]
\pump{2} &:& \pumpstattwo \dfaarrow \pumpstatthree
& \mbox{where } \pumpstatthree \subseteq \prestatone
\end{array} \right.\\
\text{Suffix analysis} &\;& \left\{
\begin{array}{rcll}\\
z  &:& \pumpstatthree \dfaarrow {\failstate}
& \mbox{where } \failstate \cap\Acc=\emptyset
\end{array} \right.
\end{array}
\end{equation*}

\begin{figure}[tp]
\begin{center}
\begin{tikzpicture}[level distance=3em]
\newcommand\failpath{ edge from parent[->] node[left] {$z$} }
\newcommand\ffail{\textrm{fail}}
\tikzstyle{level 1}=[sibling distance=25ex]
\tikzstyle{level 2}=[sibling distance=10ex]
\node {$p_{0}  $}
  child {
    node {$p_{1}$}
      child {
    node {$p_{2}$}
          child {
    node {$p_{3}$}
              child {
    node {$\ffail$}
                  edge from parent[->] node[left] {$z$}
    }
              edge from parent[->] node[left] {$\pump{}$}
    }
          edge from parent[->] node[left] {$\pump{}$}
    }
                  edge from parent[->] node[above] {$x$}
    }
  child {
  node {$\loopnode$}
    child {
      node {$\loopnode$}
    child {
      node {$\loopnode$}
        child {node {$\ffail$} \failpath } 
      edge from parent[->] node[left] {$\pump{}$}
    }
  child {
      node {$\loopnode$}
        child {node {$\ffail$} \failpath } 
      edge from parent[->] node[left] {$\pump{}$}
    }
      edge from parent[->] node[left] {$\pump{}$}
    }     
   child[missing] { }       
  child {
      node {$\loopnode$}
    child {
      node {$\loopnode$}
        child {node {$\ffail$} \failpath } 
      edge from parent[->] node[left] {$\pump{}$}
    }
  child {
      node {$\loopnode$}
        child {node {$\ffail$} \failpath } 
      edge from parent[->] node[left] {$\pump{}$}
    }
      edge from parent[->] node[left] {$\pump{}$}
    }
        edge from parent[->] node[auto] {$x$}
  } 
    child {
    node {$p_{\mathrm{Right}}$}
            edge from parent[->] node[above] {$x$}
    }
;
\end{tikzpicture}
\end{center}
\caption{Branching search tree with left context for $x\,\pump{}\,\pump\,z$}
\label{figtree}
\end{figure}
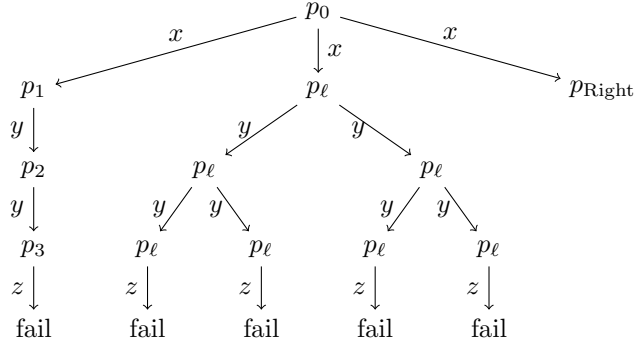

\subsection{Prefix analysis}
The analysis needs to find a string that causes the matcher to reach $\loopnode$. However, due to the nondeterminism of the underlying NFA, it is not enough to check reachability. The same string $x$ could also lead to some other states before $\loopnode$ is reached by the matcher. If one of these states could lead to acceptance, the matcher will terminate successfully, and $\loopnode$ will never be reached. In this case, there is no vulnerability, regardless of any exponential blowup in the subtree under $\loopnode$.
See \Cref{figtree}.

The REDoS prefix analysis computes all ordered multistates $\pseq{}$ reachable from $p_0$, together with a realizer $w$, using the following rules:
\[
\infer{}{(\varepsilon,p_0) \in \reachpre}
\quad
\infer{(w,\pseq{1}) \in \reachpre \quad a : \pseq 1 \ordtrans \pseq 2 \quad \not\exists w' . (w', \pseq{2}) \in \reachpre}
{\reachpre := \reachpre \cup \{(wa, \pseq{2})\}}
\] \label{fix_reachpre}
In the implementation, we keep a set $\reachpre$. It is initialized to $(\varepsilon,\startnode)$. We then repeatedly check if there is a $(w,\pseq{1})$ in the set such that for some $a$ there is a transition $a : \pseq 1 \ordtrans \pseq 2$. If there is, we add $(w\,a,\pseq 2)$ to $\reachpre$ and repeat the process. We terminate when no new $\pseq 2$ has been found in the last iteration. Finally, the analysis isolates $(w, \pseq{})$ pairs of the form $(x, \pseq{}\;\loopnode\;\pseq{}')$ and takes $\prestatone$ as $\elems{\pseq{}\;\loopnode}$ for each such pair.

\subsection{Pumping analysis}
\label{pumpinganalysis}

\begin{definition}
\label{defbranchpoint}
A branch point is a tuple:
\[
(\nondetnode, a, \{\nondeta, \nondetb\})
\]
Where $\exists \pseq{}, \pseq{}', \pseq{}'' \;.\; a : \nondetnode \mapsto \pseq{}\nondeta\pseq{}'\nondetb\pseq{}''$. We call $\nondetnode$ a nondeterministic node, in that, there are multiple successor nodes for $\nondetnode$ on the same input symbol $a$.
\end{definition}
For example, if $\nondetnode$ has three successor nodes $p_{1}$, $p_{2}$ and $p_{3}$ for the same input symbol $a$, there are three different branch points:  
\begin{align*}
(\nondetnode, a, \{p_{1}, p_{2}\})
\\
(\nondetnode, a, \{p_{1}, p_{3}\})
\\
(\nondetnode, a, \{p_{2}, p_{3}\})
\end{align*}
There can be only finitely many nondeterministic nodes in the given NFA. For each of them, we need to solve a reachability problem.

The pumping analysis can be visualized with the diagram in \Cref{pumpingmaymust}. The analysis aims to find two different paths leading from $\loopnode$ to itself. Such paths must at some point include a nondeterministic node $\nondetnode$ that has at least two transitions to different nodes $\nondeta$ and $\nondetb$ for the same symbol $a$. For such a node to lie on a path from $\loopnode$ to itself, there must be some path labeled $\pump 1$ leading from $\loopnode$ to $\nondetnode$, and moreover there must be paths from the two child nodes $\nondeta$ and $\nondetb$ leading back to $\loopnode$, such that both these paths have the label $\pump 2$. The left side of \Cref{pumpingmaymust} depicts this situation.

So far we have only considered what states \emph{may} be reached. Due to the nondeterminism of the transition relation $\nfamor apq$, there may be other states that can be reached for the same strings $\pump 1$ and $\pump 2$. Therefore, we also need to perform a \emph{must} analysis that keeps track of all states reachable via the strings we construct. This analysis uses the transition relation $\dfaarrow$ of the power DFA between sets of NFA states. In \Cref{pumpingmaymust}, it is shown on the right-hand side.

Intuitively, we run the two transition relations in parallel on the same input string. More formally, this involves constructing a product of two relations. Before we reach the branching point, we run the relations $\nfatrans$ and $\dfaarrow$ in parallel. After the nondeterministic node $\nondetnode$ has produced two different successors, we need to run two copies of $\nfatrans$ in parallel with $\dfaarrow$. One may visualize this situation by reading the diagram in \Cref{pumpingmaymust} horizontally: above the splitting at $\nondetnode$, there are two arrows in parallel for $\pump 1$, whereas below that node, there are three arrows in parallel for $a$ and $\pump 2$.

The twofold transition relation  $\twoarrow$ for running $\nfatrans$ in parallel with $\dfaarrow$ is given by the rules in \Cref{productNP}. Analogously, the threefold product transition relation $\threearrow$ for running two copies of $\nfatrans$ in parallel with $\dfaarrow$ is given by the rules in 
\Cref{productNNP}. 


In summary, the pumping analysis consists of two phases:
\begin{enumerate}
\item Given $\loopnode$ and $\prestatone$, the analysis searches for a realizer $\pump{1}$ for reaching some nondeterministic node $\nondetnode$:
\[
\twomor{\pump 1}{(\loopnode,\prestatone)}{(\nondetnode,\pumpstatone)}
\]
\item Given the successor nodes $\nondeta$ and $\nondetb$ of some $\nondetnode$ node, the analysis searches for a realizer $\pump 2$ for reaching $\loopnode$:
\[
\threemor{\pump 2}{(\nondeta,\nondetb,\pumpstattwo)}{(\loopnode,\loopnode,\pumpstatthree)}
\]
Moreover, the analysis checks that the constructed state $\pumpstatthree$ satisfies the inclusion:
\[
\pumpstatthree \subseteq \prestatone
\]
\end{enumerate}

\label{def_pumpable_nfa}If both phases of the analysis succeed, the string $\pump 1\,a\,\pump 2$ is returned as the pumpable string, together with the state $\pumpstatthree$. Moreover, we call the string $\pump 1\,a\,\pump 2$ to be a pumpable string \textit{for} the node $\loopnode$, an NFA containing such a node is then called a pumpable NFA.

\begin{figure}
\begin{center}
\begin{tikzpicture}
\matrix (m) [matrix of math nodes, row sep=2em, column sep=2em] 
{ 
{} & \loopnode &  {} &  {} & \prestatone
\\ 
{} & \nondetnode & {} & {} & \pumpstatone
\\
 \nondeta & {} & \nondetb & 
 & \pumpstattwo
\\
{} & \loopnode & {}  & {} & \pumpstatthree
\\
}; 
\path[->>,font=\normalsize] 
(m-1-2) edge node[auto] {$\pump{1}$} (m-2-2)
(m-2-2) edge node[above] {$a$} (m-3-1)
(m-2-2) edge node[above] {$a$} (m-3-3)
(m-3-1) edge node[below] {$\pump{2}\ $} (m-4-2)
(m-3-3) edge node[auto] {$\pump{2}$} (m-4-2)
; 
\path[->,ultra thick,font=\normalsize]
(m-1-5) edge node[auto] {$\pump{1}$} (m-2-5)
(m-2-5) edge node[auto] {$a$} (m-3-5)
(m-3-5) edge node[auto] {$\pump{2}\ $} (m-4-5)
; 
\end{tikzpicture}
\end{center}
\caption{Pumping analysis construction of $\pump 1\,a\,\pump 2$: ``may'' on the left using $\twoheadrightarrow$, and ``must''  on the right using $\Rrightarrow$ }
\label{pumpingmaymust}
\end{figure}
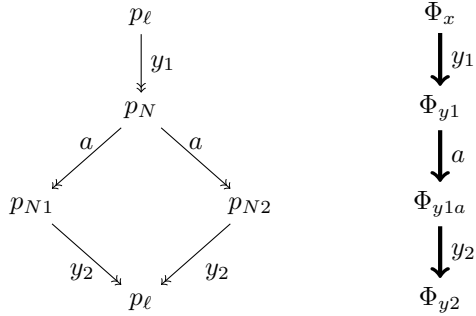

\begin{figure}[tp] 
\[
\infer{
\twomor w{(p_{1},\mstat 1)}{(p_{2},\mstat 2)}
\qquad
\begin{array}{rl}
\nfamor b{&p_{2}}{p_{3}}
\\[1ex]
\dfamor b{&\mstat 2}{\mstat 3}
\\[1ex]
\end{array}
}{
\twomor {(w\,b)}{(p_{1},\mstat 1)}{(p_{3},\mstat 3)}
}
\]
\vspace{1em}
\[
\infer{}
{
\twomor {\varepsilon}{(p_{},\mstat{})}{(p,\mstat{})}
}
\]
\caption{The twofold product transition relation $\twoarrow$}
\label{productNP}
\end{figure}

\begin{figure}[tp] 
\[
\infer{
\threemor w{(p_{1},p'_{1},\mstat 1)}{(p_{2},p'_{2},\mstat 2)}
\qquad
\begin{array}{lr}
\nfamor b{&p_{2}}{p_{3}}
\\[1ex]
\nfamor b{&p'_{2}}{p'_{3}}
\\[1ex]
\dfamor b{&\mstat 2}{\mstat 3}
\end{array}
}{
\threemor {(w\,b)}{(p_{1},p'_{1},\mstat 1)}{(p_{3},p'_{3},\mstat 3)}
}
\]
\vspace{1em}
\[
\infer{}
{
\threemor {\varepsilon}{(p_{},p'_{},\mstat{})}{(p,p',\mstat{})}
}
\]
\caption{The threefold product transition relation $\threearrow$}
\label{productNNP}
\end{figure}


\begin{example}\label{analysis:example_pumpable}
The following diagram shows an NFA corresponding to the regular expression $\kleene{(a|b|ab)}$:
\begin{center}
\begin{tikzpicture}[scale=0.8, every node/.style={transform shape},->,>=stealth',shorten >=1pt,auto,node distance=2.5cm, semithick]
\node[state, accepting] (A) at (0,0) {$p_1$};
\node[state] (B) at (2,-2) {$p_2$};

\path
  (A) edge [loop above] node {a} (A)
  (A) edge [loop left] node {b} (A)
  (A) edge [bend right] node {a} (B)
  (B) edge [bend right] node [above] {b} (A);

\end{tikzpicture}
\end{center}
Taking $\loopnode = p_1$ and $\mstat{x} = \{p_1\}$, the pumping analysis leads to the following derivation:
\begin{align*}
y_1 = \empstring \qquad &\twomor{\empstring}{(p_1, \{p_1\})}{(p_1, \{p_1\})}\\
(p_1, a, \{p_{1}, p_{2}\}) \qquad &\dfamor{a}{\{p_1\}}{\{p_1, p_2\}}\\
y_2 = b \qquad &\threemor{b}{(p_1, p_2, \{p_1, p_2\})}{(p_1, p_1, \{p_1, p_2\})}
\end{align*}
Here we have an \emph{unstable} derivation since $\{p_1, p_2\} \not\subseteq \{p_1\}$ (i.e. $\mstat{y_2} \not\subseteq \mstat{x}$). If we were to take $\mstat{x} = \{p_1, p_2\}$ (i.e. $x = a$), the resulting derivation would be stable (for the same pumpable string $ab$). Stable derivations ensure that multiple pumpings of a pumpable string do not diverge $\mstat{y_2}$, which in turn ensures the correctness of the failure suffix. We treat this inclusion more formally in Lemma~\ref{lem_stability}.
\end{example}

\subsection{Suffix analysis}
\label{suffixanalysis}
%

For each $\pumpstatthree$ constructed by the pumping analysis, the REDoS failure analysis computes all  multistates  $\suffstat$ such that there is a $z$ with:
\[
 z: \pumpstatthree \dfaarrow \suffstat
\qquad
\wedge
\qquad
\suffstat \cap\Acc=\emptyset
\]

Intuitively, $z$ \emph{fails all} the states in $\pumpstatthree$ by taking them to $\suffstat$, which does not contain any accepting states. 

\section{Test cases for the REDoS analysis} 
\label{examples}

In order to demonstrate the behavior of the analyser, here we present examples that exercise the most important aspects of its operation.

\subsection{Non commutativity of alternation}
This aspect of the analysis can be illustrated with the following two example expressions:
\[
\kleene{.} \mid \kleene{(a \mid b \mid ab)}c
\]
\[
\kleene{(a \mid b \mid ab)}c \mid \kleene{.}
\]
Even though the two expressions correspond to the same language, only the second expression yields a successful attack. In the first expression, all the multi-states starting from $\Phi_{x}$ ($\elems{\beta\loopnode}$) consist of a state corresponding to the expression $(\kleene{.})$, which implies that this expression is capable of consuming any input string thrown at it without invoking the vulnerable Kleene expression. On the other hand, $\Phi_{x}$ calculated for the second expression lacks a state corresponding to $(\kleene{.})$, leading to the following attack string configuration:
\[
x = \varepsilon \quad y = ab \quad z = \varepsilon
\]

\subsection{Prefix construction}\label{tcases_prefixcon}
Prefix construction plays one of the most crucial roles in finding an attack string. In the following example, only a certain prefix leads to a successful attack string derivation:
\[
c\kleene{.}|(c \mid d)\kleene{(a \mid b \mid ab)}e
\]
Notice that a prefix $c$ would trigger the $(\kleene{.})$ on the left due to the left-biased treatment of alternation in backtracking matchers. The prefix $d$ on the other hand forces the matcher out of this possibility. The difference between these two prefixes is captured in two different values of $(x, \Phi_{x})$:
\[
(c, \{p_1, p_2\}) \quad (d, \{p_2\})
\]
Where $p_1$ corresponds to the sub-expression $(\kleene{.})$ and $p_2$ corresponds to $\kleene{(a \mid b \mid ab)}e$. Only the latter of these two leads to a successful attack string:
\[
x = d \quad y = ab \quad z = \varepsilon
\]

Prefix construction may also lead to loop unrolling when necessary. For an example, consider the following regex:
\[
(a \mid b)\kleene{.}|\kleene{c}\kleene{(a \mid ab \mid b)}d
\]
Without the unrolling of the Kleene expression $\kleene{c}$, any pumpable string intended for the vulnerable Kleene expression will be consumed by the alternation on the left. The analyser captures this situation again as two different values of $(x, \Phi_{x})$, one for $x = c$ and the other for either $x = a$ or $x = b$. Only the former value leads to a successful attack string:
\[
x = c \quad y = ab \quad z = \varepsilon
\]
The amount of loop unrolling is limited by the finiteness of the $\Phi_x$ values. In the following example, the loop $\kleene{c}$ needs to be unrolled twice:
\[
(c \mid a \mid b)(a \mid b)\kleene{.}|\kleene{c}\kleene{(a \mid b \mid ab)}d
\]
Here, unrolling $\kleene{c}$ 0 - 2 times leads to three distinct values of $\Phi_x$ due to the different matching states on the left alternation. Only one of those unrollings leads to a successful attack string:
\[
x = cc \quad y = ab \quad z = \varepsilon
\]

\subsection{Pumpable construction}
As is the case with prefixes, the existence of an attack string may depend on the construction of an appropriate pumpable string. For an example, consider the following regex:
\[
\kleene{(a \mid a \mid b \mid b)}(a\kleene{.} \mid c)
\]
Here the pumpable string $a$ does not yield an attack string since it also triggers the $(\kleene{.})$ continuation. On the other hand, the pumpable string $b$ avoids this situation and leads to the following attack string configuration:
\[
x = \varepsilon \quad y = b \quad z = \varepsilon
\]
Similar to the prefix analysis, pumpable analysis utilises $(y, \Phi_{y})$ values to select between pumpable strings.

In some cases, the pumpable construction overlaps with prefix construction. In the example below, an attack string may be composed in two different ways:
\[
d\kleene{.}|\kleene{((c \mid d)(a \mid a))}b
\]
Here, choosing $ca$ as the pumpable string leads to a successful attack string derivation:
\[
x = \varepsilon \quad y = ca \quad z = \varepsilon
\]
However, it is also possible to form an attack string with the following configuration:
\[
x = ca \quad y = da \quad z = \varepsilon
\]
The important point here is that the attack string must begin with a $c$ instead of a $d$ in order to avoid the obvious match on the left. The analyser is capable of finding both the configurations that meet this requirement.

Pumpable construction may also lead to loop unrolling when necessary, as demonstrated by the following example:
\[
a\kleene{.}|\kleene{(\kleene{c}a(b \mid b))}d
\]
Without unrolling the inner loop $\kleene{c}$, the pumpable string $ab$ would trigger the alternation on the left. A successful attack string requires the unrolling of this inner loop, as in the following configuration:
\[
x = \varepsilon \quad y = cab \quad z = \varepsilon
\]
As with the previous example, the unrolling of the inner loop $\kleene{c}$ may be performed as part of the prefix construction, leading to the following alternate attack string configuration:
\[
x = cab \quad y = ab \quad z = \varepsilon
\]
The latter configuration may be considered more desirable in that it makes the the pumpable string shorter, leading to much smaller attack strings.


\section{Soundness of the analysis}\label{soundness}
The \machname{} performs a depth-first search of a search tree. Proofs about runs of the machine are thus complicated by the fact that the construction of the tree and its traversal are conflated. To make reasoning more compositional, we define a substructural calculus for constructing search trees. Machine runs correspond to paths from roots to leaves in these trees.

\subsection{Search tree logic}\label{treelogic}
\begin{definition}[Search tree logic]
The search tree logic has judgements of the form
\[
\treemor w{\pseq{1}}{\pseq{2}}
\]
where $w$ is an input string, and both $\pseq{1}$ and $\pseq{2}$ are sequences of NFA states. The inference rules are given in \Cref{treemor}.
\end{definition}

Intuitively, the judgement 
\[
\treemor w{\pseq{1}}{\pseq{2}}
\]
means that there is a horizontal slice of the search tree, such that the nodes at the top form the sequence ${\pseq{1}}$, the nodes at the bottom form the sequence ${\pseq{2}}$, and all paths have the same sequence of labels, forming $w$: 
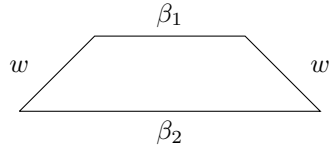
\begin{center}
\begin{tikzpicture}
\draw (0,0) -- (1,1) --(3,1) -- (4,0)-- (0,0) ;
\node[above] at (2,1) {$\pseq 1$};
\node[] at (0,.6) {$w$};
\node[] at (4,.6) {$w$};
\node[below] at (2,0) {$\pseq 2$};
\end{tikzpicture}
\end{center}
Each $w$ represents an NFA run $\nfamor w{p_{1}}{p_{2}}$ for some $p_{1}$ that occurs in $\pseq 1$ and some $p_{2}$ that occurs in $\pseq 2$. The string $w$ labels the sides of the trapezoid, since that determines the compatible boundary for parallel composition. Again we may like to think of $w$ as a proof of reachability. Here the reachability is not in the NFA, but in the matcher based on it.

\begin{figure}[tp]
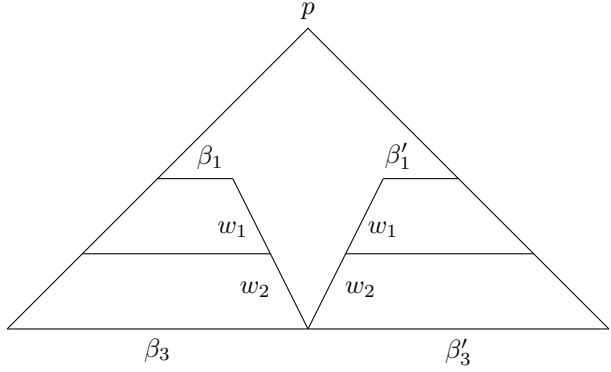

\[
\infern{\omor{a}{p}{\pseq{}}}{\treemor{a}{p}{\pseq{}}}{\textsc{Trans}}
\]
\[
\infern{\treemor{w_{1}}{\pseq{1}}{\pseq{2}} \qquad \treemor{w_{2}}{\pseq{2}}{\pseq{3}}}
{\treemor{(w_{1}\,w_{2})}{\pseq{1}}{\pseq{3}}}{\textsc{SeqComp}}
\]
\[
\infern{}{\treemor \varepsilon{\pseq{}}{\pseq{}}}{\varepsilon\textsc{Seq}}
\]
\[
\infern{\treemor{w}{\pseq{1}}{\pseq{2}} \qquad \treemor{w}{\pseq{1}'}{\pseq{2}'}}
{\treemor{w}{(\pseq{1}\,\pseq{1}')}{(\pseq{2}\,\pseq{2}')}}{\textsc{ParComp}}
\]
\[
\infern{}{\treemor w\varepsilon\varepsilon}{\varepsilon\textsc{Par}}
\]
\caption{Search tree logic}
\label{treemor}
\end{figure}


The trapezoid can be stacked on top of each other if they share a common $\pseq{}$ at the boundary. They can be placed side-by-side if they have the same $w$ on the inside:

\begin{center}
\begin{tikzpicture}
\draw (0,0) -- (4,4) --(8,0);
\draw (0,0) -- (8,0);
\draw (2,2) -- (3,2);
\draw (5,2) -- (6,2);
\draw (1,1) -- (3.5,1);
\draw (4.5,1) -- (7,1);
\draw (3,2) -- (4,0);
\draw (5,2) -- (4,0);

\node[above] at (3,1.1) {$w_{1}$};
\node[above] at (5,1.1) {$w_{1}$};

\node[above] at (3.3,0.3) {$w_{2}$};
\node[above] at (4.7,0.3) {$w_{2}$};

\node[above] at (4,4) {$p$};
\node[above] at (2.7,2) {$\beta_{1}$};

\node[above] at (5.2,2) {$\beta_{1}'$};
\node[below] at (6,0) {$\beta_{3}'$};
\node[below] at (2,0) {$\beta_{3}$};
\end{tikzpicture}
\end{center}
Note that the search tree logic in Figure~\ref{treemor} is substructural, in that, common structural rules such as weakening and contraction do not hold in general:
\[
\infern{\treemor{w}{\pseq{1}}{\pseq{1}'}}{\treemor{w}{\pseq{1}\pseq{2}}{\pseq{1}'}}{\textsc{NotWeaken}} \quad \infern{\treemor{w}{\pseq{}\pseq{}}{\pseq{}'}}{\treemor{w}{\pseq{}}{\pseq{}'}}{\textsc{NotContract}}
\]

\subsection{Pumpable implies exponential tree growth}
We use the search tree logic to construct a tree by closely following the phases of our REDoS analysis. The exponential growth of the search tree in response to pumping is easiest to see when thinking of horizontal slices across the search tree for each pumping of $\pump{}$. The machine computes a diagonal cut across the search tree as it moves towards the left corner. The analysis constructs horizontal cuts with all states at the same depth. It is sufficient to show that the width of the search tree grows exponentially. The width is easier to formalize than the size.

We need a series of technical lemmas connecting different transition relations.
\begin{lemma}\label{dfamorunion}
The following rule is admissible:
\[
\infer{\dfamor w{\mstat{1}}{\mstat{2}} \qquad \dfamor w{\mstat{1}'}{\mstat{2}'}}
{\dfamor w{(\mstat{1} \cup \mstat{1}')}{(\mstat{2}\cup \mstat{2}')}}
\]
\end{lemma}

\begin{lemma}[$\dfaarrow\triangle$ simulation]\label{dfatree}
If $\dfamor w{\mstat 1}{\mstat 2}$, $\treemor w{\pseq 1}{\pseq 2}$ and $\mstat 1= \elems{\pseq 1}$, then $\mstat 2= \elems{\pseq 2}$.
\end{lemma}
\begin{proof}
Suppose:
\[
\pseq 1 = (p_1 \ldots p_n) \qquad a : p_i \mapsto \theta_i
\]
Then from the search tree logic we get $\treemor {a}{\pseq 1}{(\theta_1 \ldots \theta_n)}$. Moreover, the definition of $\dfaarrow$ implies $\dfamor a{\{p_i\}}{\elems{\theta_i}}$. Now, applying Lemma~\ref{dfamorunion} gives:
\[
\dfamor{a}{\elems{\pseq 1}}{\elems{\theta_1} \cup \ldots \cup \elems{\theta_n}} = \elems{\theta_1 \ldots \theta_n}
\]
Therefore, the result holds for strings of unit length. An induction on the length of $w$ completes the proof.
\end{proof}

\begin{lemma}[$\ordtrans\triangle$ simulation]\label{ordtranstree}
If $w:\pseq 1 \ordtrans \pseq 2$, $\treemor w{\pseq 1'}{\pseq 2'}$ and $\pseq 1' \ggg \pseq 1$, then $\pseq 2' \ggg \pseq 2$. 
\end{lemma}
\begin{proof}
Suppose:
\[
\pseq{1}' = (p_{11} \ldots p_{mk}) \qquad a : p_{ij} \mapsto \theta_{ij}
\]
Where $p_{ij}$ corresponds to the $j$th occurrence of the state $p_i$. Equivalently:
\[
p_{ij} = p_{i'j'} \Longleftrightarrow i = i'
\]
Given $\pseq{1}' \ggg \pseq{1}$, we deduce:
\[
(p_{11} \ldots p_{mk}) \ggg (p_{11} \ldots p_{m1}) = \pseq{1}
\]
Now, given $a : \pseq{1} \ordtrans \pseq{2}$, the definition of $\ordtrans$ gives:
\[
(\theta_{11} \ldots \theta_{m1}) \ggg \pseq{2}
\]
On the other hand, $\treemor{a}{\pseq{1}'}{\pseq{2}'}$ gives:
\[
\pseq{2}' = (\theta_{11} \ldots \theta_{mk})
\]
The definition of $\ggg$ can be generalized for multi-states, which leaves us with:
\[
(\theta_{11} \ldots \theta_{mk}) \ggg (\theta_{11} \ldots \theta_{m1})
\]
That is, we have shown:
\[
\pseq{2}' = (\theta_{11} \ldots \theta_{mk}) \ggg (\theta_{11} \ldots \theta_{m1}) \ggg \pseq{2}
\]
An induction on the length of $w$ completes the proof.
\end{proof}

\begin{lemma}[$\nfatrans \triangle$ simulation]\label{nfatotree}
Given $\nfamor wpq$, there are sequences of states $\pseq 1$ and $\pseq 2$ such that $\treemor wp{{\pseq{1}}\,q\,{\pseq{2}}}$.
\end{lemma}
\begin{proof}
The base case ($w = a$) holds from the definition of $\triangle$. For the inductive step, suppose $\nfamor{w}{p}{q}$ and $a : q \mapsto q'$. Then from the induction hypothesis we get $\treemor{w}{p}{\pseq{1}\,q\,\pseq{2}}$ for some $\pseq{1}, \pseq{2}$. Moreover, from the base case we have $\treemor{a}{q}{\pseq{3}\,q'\,\pseq{4}}$ for some $\pseq{3}, \pseq{4}$. Assuming $\treemor{a}{\pseq{1}}{\pseq{1}'}$ and $\treemor{a}{\pseq{2}}{\pseq{2}'}$ for some $\pseq{1}', \pseq{2}'$, the definition of $\triangle$ gives $\treemor{wa}{p}{\pseq{1}'\,\pseq{3}\,q'\,\pseq{4}\,\pseq{2'}}$.
\end{proof}

\begin{lemma}[Pumpable realizes non-linearity] \label{lem_pump_nonlin}
Let $y$ be pumpable  for some node $\loopnode$. Then there exist $\pseq 1$, $\pseq 2$, $\pseq 3$ such that:
\[
\treemor y{\loopnode}{\pseq 1\,\loopnode\,\pseq 2\,\loopnode\,\pseq 3}
\]
\end{lemma}
\begin{proof}
The pumpable analysis generates a string of the form:
\[
\pump{} = \pump 1\,a\,\pump 2
\]
Where
\[
\nfamor{\pump{1}}{\loopnode}{\nondetnode}
\]
\[
a : \nondetnode \mapsto (\pseq{}\;\nondeta\;\pseq{}'\;\nondetb\;\pseq{}'')
\]
\[
\nfamor{\pump{2}}{\nondeta}{\loopnode}
\qquad
\nfamor{\pump{2}}{\nondetb}{\loopnode}
\]
Now, Lemma~\ref{nfatotree} leads to the desired result.
\end{proof}

\begin{lemma}\label{lem_stability}
Let $x, y$ be constructed from the prefix analysis and the pumpable analysis such that:
\[
\begin{array}{rcll}
x &:& \startnode \ordtrans (\pseq{}\;\loopnode\;\pseq{}')\\
y &:& \elems{\pseq{}\,\loopnode} \dfaarrow \mstat{y}
& \qquad \mstat{y} \subseteq \elems{\pseq{}\,\loopnode}
\end{array}
\]
Then the following holds for any natural number $n$:
\[
\mstat{y^n} \subseteq \mstat{y^{n - 1}}
\]
Where $\mstat{y^0} = \elems{\pseq{}\,\loopnode}$ and $y^n : \mstat{y^0} \dfaarrow \mstat{y^n}$.
\end{lemma}
\begin{proof}
By induction on $n$. Note that the base case ($n = 1$) holds by construction. For the inductive step, suppose $\exists \; q \in \mstat{y^n}$, then from the definition of $\mstat{y^n}$ we get $\exists \; p \in \mstat{y^{n - 1}} \;.\;  y : p \nfatrans q$. Moreover, the induction hypothesis gives $\mstat{y^{n - 1}} \subseteq \mstat{y^{n - 2}}$. Therefore, we have $p \in \mstat{y^{n - 2}}$, which in turn implies $q \in \mstat{y^{n - 1}}$.
\end{proof}

The importance of Lemma~\ref{lem_stability} is that it allows us to calculate a failure suffix $z$ independent of the number of pumping iterations; $\mstat{y^n}$ can only shrink as $n$ increases.

\begin{lemma}[Exponential tree growth]\label{lem_expn_growth}
Let $x, y, z$ be constructed from the analysis such that:
\[
\begin{array}{rcll}
x &:& \startnode \ordtrans (\pseq{}\;\loopnode\;\pseq{}')\\
y &:& \elems{\pseq{}\,\loopnode} \dfaarrow \mstat{y}
& \qquad \mstat{y} \subseteq \elems{\pseq{}\,\loopnode}\\
z &:& \mstat{y} \dfaarrow {\failstate}
& \qquad \failstate \cap\Acc=\emptyset
\end{array}
\]
Then there exists $\pseq{L}, \pseq{R}$ such that:
\[
\treemor{x}{\startnode}{\pseq{L}\,\loopnode\,\pseq{R}} \,\land\, \elems{\pseq{L}} = \elems{\pseq{}} \tag{A}
\]
\[
\treemor{y^n}{\pseq{L}\,\loopnode}{\pseq{n}} \Rightarrow \length{\pseq{n}} \ge 2^n \tag{B}
\]
\[
z : \elems{\pseq{n}} \dfaarrow \mstat{}' \Rightarrow \mstat{}' \cap \Acc = \emptyset \tag{C}
\]
\end{lemma}
\begin{proof}
\begin{itemize}
\item
\textbf{Statement (A)}: Suppose $\treemor{x}{\startnode}{\pseq{x}}$. Then from Lemma~\ref{ordtranstree} it follows that $\pseq{x} \ggg \pseq{}\,\loopnode\,\pseq{}'$. That is, $\loopnode$ must occur in $\pseq{x}$. If we dissect $\pseq{x}$ into $\pseq{L}\,\loopnode\,\pseq{R}$ such that $\loopnode \not\in \elems{\pseq{L}}$, then from the definition of $\ggg$ it follows that $\elems{\pseq{L}} = \elems{\pseq{}}$.
\item
\textbf{Statement (B)}: Follows from Lemma~\ref{lem_pump_nonlin}. The number of copies of $\loopnode$ doubles at each pumping iteration.
\item
\textbf{Statement (C)}: Suppose $y^n : \elems{\pseq{L}\,\loopnode} \dfaarrow \pseq{n}'$ . Since $\elems{\pseq{L}\,\loopnode} = \elems{\pseq{}\,\loopnode}$ (statement A), Lemma~\ref{dfatree} gives: $\elems{\pseq{n}} = \elems{\pseq{n}'}$. Now from Lemma~\ref{lem_stability} it follows that $\elems{\pseq{n}} \subseteq \mstat{y}$. Since $z$ cannot lead to a successful match from any state in $\mstat{y}$ (by construction), the same should be true for $\elems{\pseq{n}}$.
\end{itemize}
\end{proof}

Lemma~\ref{lem_expn_growth} may be visualized as in Figure~\ref{fig_expn_growth}. Note that the right hand slice of the tree (emanating from $\pseq{}'$) is irrelevant, the depth-first strategy of a backtracking matcher forces it to explore the left hand slice first. Since none of the states at the bottom of the tree ($\pseq{n}'$) are accepting, it is forced to explore the (exponentially large, $\length{\pseq{n}} \ge 2^n$) entire slice (as proved in the following section).

\begin{figure}
\begin{center}
\begin{tikzpicture}
\draw (0,0) -- (4,4);
\draw[dashed] (4,4) -- (4,0);
\draw (4,0) -- (0, 0);
\node[above] at (4,4) {$p_0$};

\draw[dashed] (4,4) -- (6,0) -- (4, 0);

\draw (3,3) -- (4,3);
\node[below] at (3.5,3) {$\pseq{}$};

\draw[->, line width=1pt] (5,4) to (4.1,3.1);
\node[above right] at (5,4) {$\loopnode$};

\draw[dashed] (4,3) -- (4.5, 3);
\node[below] at (4.25,3) {$\pseq{}'$};

\draw (1,1) -- (4,1);
\node[below] at (2.5, 1) {$\pseq{n}$};
\node[below] at (2, 0) {$\pseq{n}'$};

\draw[|-|] (-1, 0) -- (-1, 1);
\node[left] at (-1, 0.5) {$z$};

\draw[-|] (-1, 1) -- (-1, 3);
\node[left] at (-1, 2) {$y^n$};

\draw[-|] (-1, 3) -- (-1, 4);
\node[left] at (-1, 3.5) {$x$};
\end{tikzpicture}
\caption{Tree growth (Lemma~\ref{lem_expn_growth})}
\label{fig_expn_growth}
\end{center}
\end{figure}
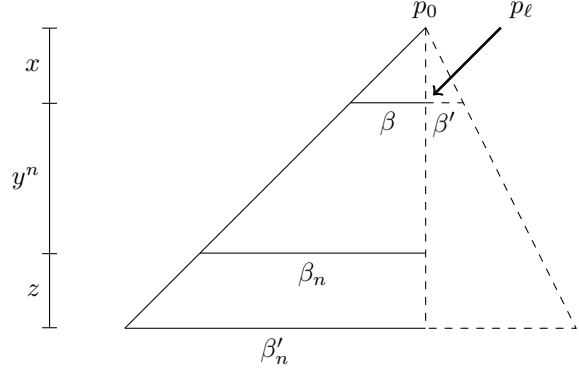

\newcommand\pin{\bar{\in}}
\subsection{From search tree to machine runs}\label{sec_treetoruns}
Having proved that the attack strings lead to exponentially large search trees, in this section we show how backtracking matchers are forced to traverse all of it. We use the notation $w[i:j]$ to represent the substring of $w$ starting at index $i$ (inclusive) and ending at index $j$ (exclusive). That is,
\begin{align*}
w[i:i] &= \empstring\\
w[i:j] &= w[i]...w[j-1] \qquad (i < j)
\end{align*}
We also use the notation $q \pin \pseq{}$ to identify an appearance of a state $q$ within $\pseq{}$:
\[
q \pin \pseq{} \Leftrightarrow \exists \pseq{}', \pseq{}'' \;.\; \pseq{} = \pseq{}' q \pseq{}''
\]

\begin{lemma}\label{lem_intermediate_reject}
Let $w$ be an input string of length $n$, $s$ a (constant) offset into $w$ ($0 \le s < n$) and $p$ a state such that:
\[
\treemor{w[s:i]}{p}{\pseq{i}} \quad s \le i \le n 
\]
\[
\elems{\pseq{n}} \cap \Acc = \emptyset
\]
Then for any $q \pin \pseq{i}$, and for any $\piseqs{}$, the following run exists:
\[
w \Vdash (q, i)\piseqs{} \cstep{*} \piseqs{}
\]
\end{lemma}
\begin{proof}
By induction on $(n - i)$. For the base case $(i = n)$, we have the machine state:
\[
(q, n)\,\piseqs{}
\]
Since $q \pin \pseq{n}$, this is not an accepting configuration. Therefore, we have:
\[
w \Vdash (q, n)\,\piseqs{} \cstep{} \piseqs{}
\]
For the inductive step, suppose $i = k\;(s \le k < n)$, then we have the machine:
\[
w \Vdash (q, k)\,\piseqs{}
\]
If $q$ has no transitions on $w[k]$, the proof is trivial. Let us assume:
\[
w[k] : q \mapsto q'_0 \ldots q'_m 
\]
Then we have the transition:
\[
w \Vdash (q, k)\,\piseqs{} \cstep{} (q'_0, k + 1) \ldots (q'_m, k + 1) \, \piseqs{}
\]
Now the definition of $\triangle$ implies that $q'_0, \ldots, q'_m$ are part of $\pseq{k + 1}$. Therefore, we can apply the induction hypothesis to each of the newly spawned frames in succession, which leads to the desired result.
\end{proof}

\begin{lemma}\label{lem_tree_cover}
Let $w$ be an input string of length $n$, $s$ a (constant) offset into $w$ ($0 \le s < n$) and $p$ a state such that:
\[
\treemor{w[s:i]}{p}{\pseq{i}} \quad s \le i \le n
\]
\[
\elems{\pseq{n}} \cap \Acc = \emptyset
\]
Then for any $q \pin \pseq{i}$, the following run exists (for some $\piseqs{}$):
\[
w \Vdash (p, s) \cstep{*} (q, i)\piseqs{}
\]
\end{lemma}
\begin{proof}
By induction on $(i - s)$. The base case $(i = s)$ holds trivially. For the inductive step, suppose  $i = k \; (s < k \le n)$ and that $\dot{q} \pin \pseq{k}$. Then from the definition of $\triangle$, there must be some $q' \pin \pseq{k - 1}$ such that:
\[
\pseq{k - 1} = \pseq{}\,q'\,\pseq{}' \qquad w[k - 1] : q' \nfatrans q_0 \ldots \dot{q} \ldots q_m 
\]
Now from the induction hypothesis we get:
\[
w \Vdash (p, s) \cstep{*} (q', k - 1)\,\piseqs{}
\]
Therefore, we deduce the run:
\[
w \Vdash (p, s) \cstep{*} (q', k - 1)\piseqs{} \cstep{} (q_0, k) \ldots (\dot{q}, k) \ldots (q_m, k)\,\piseqs{}
\]
At this point, applying Lemma~\ref{lem_intermediate_reject} to the newly spawned frames yields the required result.
\end{proof}

Given a search tree with all failure nodes at the bottom, Lemma~\ref{lem_intermediate_reject} shows that any intermediate frame reached during a simulation will eventually be rejected. Moreover, Lemma~\ref{lem_tree_cover} shows that a simulation corresponding to such a search tree is forced to visit each and every node of the tree.

\begin{lemma}[Tree traversal]\label{lem_tree_traverse}
Suppose $w$ is an input string of length $n$ such that:
\begin{align*}
\treemor{w[0:s] &}{p_0}{\pseq{}\,\overline{\pseq{}}}\\
\treemor{w[s:n] &}{\pseq{}}{\pseq{}'}\\
\elems{\pseq{}'}&\cap\Acc =\emptyset
\end{align*}
Then for any $q \pin \pseq{}$, the following machine run exists (for some $\piseqs{}$):
\[
w \Vdash (p_0, 0) \cstep{*} (q, s)\,\piseqs{}
\]
\end{lemma}
\begin{proof}
By induction on $s$. Note that the base case ($s = 0$) follows from Lemma~\ref{lem_tree_cover}. For the inductive step, suppose $s > 0$. Here we focus on the lowest common ancestor of all the states in $\pseq{}$, this situation is illustrated in the following figure:
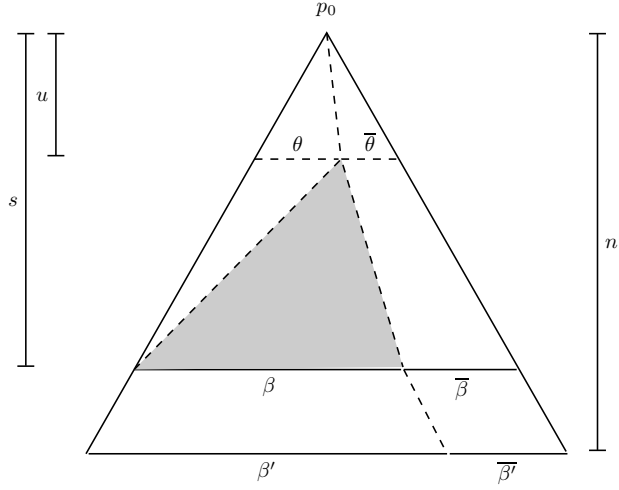
\begin{center}
\begin{tikzpicture} [scale=0.8, every node/.style={transform shape},>=stealth',shorten >=1pt,auto,node distance=2.5cm, semithick]
\coordinate (left) at (0,0);
\coordinate (right) at (8,0);
\coordinate (top) at (4,7);
\node [above =1ex of top] {$p_0$};

\coordinate (u) at ($ (top) !.3! (left) $);
\coordinate (v) at ($ (top) !.3! (right) $);
\coordinate (uv) at ($ (v) !.4! (u) $);

\coordinate (p) at ($ (left) !.2! (top) $);
\coordinate (q) at ($ (right) !.2! (top) $);
\coordinate (pq) at ($ (q) !.3! (p) $);
\draw (p)--(pq) node[pos=.5,below] {$\pseq{}$};
\draw (pq)--(q) node[pos=.5,below] {$\overline{\pseq{}}$};

\coordinate (bot) at ($ (right) !.25! (left) $);

\draw (left)--(top)--(right);
\draw (right)--(bot) node[pos=.5,below] {$\overline{\pseq{}'}$};
\draw (bot)--(left) node[pos=.5,below] {$\pseq{}'$};

\draw[dashed] (u)--(uv)node[pos=.5,above] {$\theta$};
\draw[dashed] (uv)--(v)node[pos=.5,above] {$\overline\theta$};

\draw[dashed] (top)--(uv);
\draw[dashed] (pq)--(bot);

\draw[fill=black!20,dashed] (p)--(uv)--(pq);

\draw[|-|] let \p{x} = (uv) in (-.5,7 ) -- (-.5,\y{x}) node[pos=.5,left] {$u$};
\draw[|-|] let \p{x} = (pq) in (-1,7 ) -- (-1,\y{x}) node[pos=.5,left] {$s$};
\draw[|-|] let \p{x} = (bot) in (8.5,7 ) -- (8.5,\y{x}) node[pos=.5,right] {$n$};

\end{tikzpicture}
\end{center}
Let us assume that this state (lowest common ancestor of $\pseq{}$) occurs at depth $u$ ($u > 0$, as otherwise we would have the base case again). Now from the diagram we deduce:
\begin{align*}
\treemor{w[0:u] &}{p_0}{\theta\,\overline{\theta}}\\
\treemor{w[u:n] &}{\theta}{\pseq{}'}\\
\elems{\pseq{}'}&\cap\Acc =\emptyset
\end{align*}
Therefore, if $\dot{q}$ is the lowest common ancestor of $\pseq{}$, from the induction hypothesis (since $u < s$) we get:
\[
w \Vdash (p_0, 0) \cstep{*} (\dot{q}, u)\,\piseqs{}
\]
For some $\piseqs{}$. Moreover, from Lemma~\ref{lem_tree_cover} we deduce:
\[
w \Vdash (\dot{q}, u) \cstep{*} (q, s)\,\piseqs{}'
\]
For some $q \pin \pseq{}$ and a failure continuation $\piseqs{}'$. Finally, we use Lemma~\ref{lem_backtrack_extend} to compose these two runs into:
\[
w \Vdash (p_0, 0) \cstep{*} (\dot{q}, u)\,\piseqs{} \cstep{*} (q, s)\,\piseqs{}'\,\piseqs{}
\]
\end{proof}


In sum, we have shown that the pumped part of the search tree grows exponentially in the size of the input, and that the \machname{} is forced to traverse all of it.

\begin{theorem}[Redos analyis soundness]\label{thm_analysis_soundness}
Let the strings $x$, $\pump{}$ and $z$ be constructed by the REDoS analysis. Let $k$ be an integer. Then the \machname{} takes at least $2^{k}$ steps on the input string $x\,y^{k}\,z$
\end{theorem}
\begin{proof}
Follows from Lemma~\ref{lem_expn_growth} and Lemma~\ref{lem_tree_traverse}.
\end{proof}

\section{Completeness of the analysis}\label{completeness}
The analysis assumes that only a pumpable NFA can lead to an exponential runtime vulnerability. For completeness, we need to ensure that there are no other configurations that can cause such a vulnerability. Here we show that for any non-pumpable NFA, the width of any search tree is bounded from above by a polynomial. In places where an NFA is mentioned in a discussion below, a non-pumpable NFA is to be assumed (unless otherwise mentioned).

\newcommand\pcount[2]{[#1]_#2}
\newcommand\peq[1]{\overset{#1}{\sim}}
\newcommand\psim{\simeq}




\begin{definition}
For an ordered multi-state $\pseq{}$ and a state $p$, we define the function $\pcount{\pseq{}}{p}$ as the number of ocurrences of $p$ within $\pseq{}$. Moreover, the relations $\peq{p}$ ($p$\textit{-simulate}) and $\psim$ (\textit{simulate}) on ordered multi-states are incrementally defined as follows:
\[
\pseq{} \peq{p} \pseq{}' \Leftrightarrow \pcount{\pseq{}}{p} = \pcount{\pseq{}'}{p}
\]
\[
\pseq{} \psim \pseq{}' \Leftrightarrow \forall \; p \in Q \;.\; \pseq{} \peq{p} \pseq{}'
\]
It can be shown that both $\peq{p}$ and $\psim$ are reflexive, symmetric and transitive relations.
\end{definition}

\begin{lemma}\label{lem_psim_props}
The relation $\psim$ can be shown to satisfy the following basic properties:
\[
\pseq{} \psim \pseq{}' \Rightarrow \elems{\pseq{}} = \elems{\pseq{}'} \land \length{\pseq{}} = \length{\pseq{}'}
\]
\[
\pseq{1}\pseq{2} \psim \pseq{3}\pseq{4} \Leftrightarrow \forall \pseq{}\;.\;\pseq{1}\pseq{}\pseq{2} \psim \pseq{3}\pseq{}\pseq{4}
\]
\[
\pseq{} \psim \pseq{1}\pseq{}'\pseq{2} \land \pseq{}' \psim \pseq{}'' \Rightarrow \pseq{} \psim \pseq{1}\pseq{}''\pseq{2}
\]
\[
\pseq{1} \psim \pseq{1}' \land \pseq{2} \psim \pseq{2}' \Rightarrow \pseq{1}\pseq{2} \psim \pseq{1}'\pseq{2}'
\]
\end{lemma}


\begin{lemma}\label{lem_triangle_psim}               
Let $w$ be an input string, $\pseq{1}$, $\pseq{2}$ be ordered multi-states such that:
\[
\pseq{1} \psim \pseq{2}
\quad
\treemor{w}{\pseq{1}}{\pseq{1}'}
\quad
\treemor{w}{\pseq{2}}{\pseq{2}'}
\]
Then $\pseq{1}' \psim \pseq{2}'$.
\end{lemma}
\begin{proof}
Informally, $\pseq{2}$ is merely a re-ordering of $\pseq{1}$ (and vice versa). The trapezoid emanating from $\pseq{1}$ will be composed of individual search trees rooted at each constituent state of $\pseq{1}$. Therefore, the trapezoid emanating from $\pseq{2}$ will be a re-ordering of those search trees.

More formally, let $n = \length{\pseq{1}} = \length{\pseq{2}}$ (the latter equality holds since $\pseq{1} \psim \pseq{2}$). We perform an induction on $n$. The base case ($n = 1$) follows from the definition of $\triangle$ ($\pseq{1} = \pseq{2} = p$ for some $p \in Q$). For the inductive step, note that any state $q$ introduced to both $\pseq{1}$ and $\pseq{2}$ (to make them $n + 1$ in size) must be the same (in order to preserve $\pseq{1} \psim \pseq{2}$). Since the search tree rooted at $q$ is same for both $\pseq{1}$ and $\pseq{2}$ (regardless of where it appears within each of the multi-states), its contribution to $\pseq{1}'$ and $\pseq{2}'$ is the same.
\end{proof}


\begin{definition}
Given an NFA, a path $\gamma$ is a sequence of triples:
\[
(p_{0},a_{0},p_{1}) (p_{1},a_{1},p_{2}) \ldots (p_{n - 1},a_{n - 1},p_{n})
\]
where for $0 \leq i < n$, there is a transition $a_i : p_{i} \to p_{i + 1}$ in the NFA. We write $\mathrm{dom}(\gamma)$ for the first node $p_{0}$ and $\mathrm{cod}(\gamma)$ for the last node  $p_{n}$ in the path. The sequence of input symbols $a_{0}  \ldots a_{n}$ along the path is written as $\ptos\gamma$ and called the label of the path. Moreover, the set of nodes $\{p_{0}, \ldots , p_{n}\}$ along the path $\gamma$ is written as $\gnodes{\gamma}$.
\end{definition}

\begin{lemma}\label{lem_triangle_paths}
Given a tree judgement $\treemor {w}{p}{\pseq{}}$, for any state $q$ appearing in $\pseq{}$, there exists a path $\gamma$ with $\gdom{\gamma} = p$, $\gcod{\gamma} = q$ and $\glab{\gamma} = w$.
\end{lemma}
\begin{proof}
By induction on the length of $w$.
\end{proof}

\begin{definition}\label{def_two_paths}
We write $\twopaths wpq$ ($p$ \texttt{two-paths} $q$) iff $w = aw'$ and:
\[
\exists \pseq{}, \pseq{}', \pseq{}'' \;.\; a : p \mapsto \pseq{}p_{1}\pseq{}'p_{2}\pseq{}''
\]
\[
(w' = \empstring \land p_{1} = p_{2} = q) \quad \texttt{or} \quad (\nfamor{w'}{p_{1}}{q} \land \nfamor{w'}{p_{2}}{q})
\]
\end{definition}

Definition~\ref{def_two_paths} allows us to formulate pumpability in a different notation; if we have $\twopaths{w_1}{p}{q}$ and $\nfamor{w_2}{q}{p}$ (loop), then the state $p$ is pumpable on the input string $w_1w_2$ (see Figure~\ref{pumpingmaymust}).

\begin{definition}\label{def_gshared}
Let $\gamma$ be a path. We define the sets $\gshared{\gamma}$ and $\gfree{\gamma}$ as follows:
\begin{align*}
\gshared{\gamma} &= \{ p \;|\; \exists \; \gamma_1, \gamma_2 \;.\;  \gamma = \gamma_1\gamma_2\\
&\land \twopaths{\glab{\gamma_2}}{\gdom{\gamma_2}}{\gcod{\gamma_2}} \land p = \gdom{\gamma_2} \}\\
\gfree{\gamma} &= Q \setminus \gshared{\gamma}
\end{align*}
\end{definition}

Essentially, $\gshared{\gamma}$ identifies the nondeterministic states along a path $\gamma$. There are at least two paths from a given state in $\gshared{\gamma}$ to $\gcod{\gamma}$ bearing the same label ($\glab{\gamma_2}$ above, a suffix of $\glab{\gamma}$).

\begin{lemma}\label{lem_sib_restrict}
Suppose $\gamma$ is a path corresponding to a non-pumpable NFA such that $p = \gcod{\gamma}$. Then the following holds for any $w$:
\[
\treemor{w}{p}{\pseq{}} \Rightarrow \elems{\pseq{}} \subseteq \gfree{\gamma}
\]
\end{lemma}
\begin{proof}
From the definitions we have:
\[
\elems{\pseq{}} \subseteq Q = \gshared{\gamma} \cup \gfree{\gamma}
\]
Suppose $q \in \elems{\pseq{}} \cap \gshared{\gamma}$. Then $q \in \gshared{\gamma}$ gives:
\[
\exists \; w' \;.\; \twopaths{w'}{q}{p}
\]
However, since $q \in \elems{\pseq{}}$ we also have:
\[
w : p \nfatrans q
\]
Leading to the contradiction:
\[
\twopaths{w'w}{q}{p} \nfatrans q
\]
Therefore, it must be the case that $\elems{\pseq{}} \cap \gshared{\gamma} = \emptyset$. This leads to the conclusion:
\[
\elems{\pseq{}} \subseteq \gfree{\gamma}
\]
\end{proof}

\begin{figure}
\begin{center}
\begin{tikzpicture}
\draw (0,0) -- (3,3) -- (6,0) -- (0, 0);
\draw (3,5) -- (3,3);
\node[above] at (3,5) {$p_0$};
\node[right] at (3,3.1) {$q$};

\draw[dashed] (3,3) -- (1,0);
\node[above] at (0.9,0) {$p$};

\draw[dashed] (3,3) -- (5,0);
\node[above] at (5.1,0) {$p$};

\draw (0, -1) -- (1, 0) -- (2, -1) -- (0, -1);
\node[below] at (1, -1) {$\pseq{}$};

\draw (4, -1) -- (5, 0) -- (6, -1) -- (4, -1);
\node[below] at (5, -1) {$\pseq{}$};

\draw[|-] (6.5, -1) -- (6.5, 0);
\node[right] at (6.5, -.5) {$w$};

\draw[|-|] (6.5, 0) -- (6.5, 3);
\node[right] at (6.5, 1.5) {$w'$};
\end{tikzpicture}
\caption{Sibling restriction on $\gshared{\gamma}$}
\label{fig_sib_restrict}
\end{center}
\end{figure}

\noindent
Lemma~\ref{lem_sib_restrict} is illustrated in Figure~\ref{fig_sib_restrict}. Note that the fringes of the sibling trees rooted at the two $p$'s are identical ($\triangle$ logic is deterministic), making it impossible for either of them to contain a $q$ ($q$ would be pumpable otherwise). In other words, $q \in \gshared{\gamma}$ cannot appear again within search trees rooted at $\gcod{\gamma}$. It is this restriction on nondeterminism that leads us to the polynomial bound. However, flushing out this polynomial bound requires quite an elaborate analysis of the search tree structure, as we shall see next.

\newcommand\iseq[1]{\mathring{#1}}
\newcommand\sseq[1]{\sigma_{#1}}
\newcommand\bseq[1]{\alpha_{#1}}

\newcommand\pgrp{\triangleright}
\newcommand\pgrpx{\triangleright\triangleright}
\begin{definition}
We define the reduction $\pgrp$ on pairs of ordered multi-states according to the following rules:
\[
(q_1 \ldots q_n, \pseq{1}q\pseq{2}) \pgrp (q_1 \ldots q_iq \ldots q_n, \pseq{1}\pseq{2}) \quad (\exists \;i \;.\; q = q_i)
\]
\[
(q_1 \ldots q_n, \pseq{1}q\pseq{2}q\pseq{3}) \pgrp (q_1 \ldots q_n qq, \pseq{1}\pseq{2}\pseq{3}) \quad (\not\exists \; i \;.\; q = q_i)
\]
The reduction $\pgrp$ groups repeated states together. Given that each transition decreases the length of the second component, the reduction must terminate. We use the notation $\pgrpx$ to denote a maximal reduction:
\[
(\bseq{1}, \pseq{1}) \pgrpx (\bseq{2}, \pseq{2}) \Rightarrow \not\exists (\bseq{3}, \pseq{3}) \;.\; (\bseq{2}, \pseq{2}) \pgrp (\bseq{3}, \pseq{3})
\]
\end{definition}

\begin{lemma}\label{lem_pgrpx_props}
For a reduction $(\empseq, \pseq{}) \pgrpx (\bseq{}, \sseq{})$, the following basic properties can be shown to hold:
\[
\pseq{} \psim \bseq{}\sseq{} \tag{a}
\]
\[
\elems{\bseq{}} \cup \elems{\sseq{}} = \elems{\pseq{}} \tag{b}
\]
\[
\forall \; p \in \elems{\bseq{}} \;.\; \pcount{\bseq{}}{p} = \pcount{\pseq{}}{p} > 1 \tag{c}
\]
\[
\length{\sseq{}} = \card{\elems{\sseq{}}} \tag{d}
\]
\end{lemma}

\newcommand\gtriangle{\bar{\bigtriangleup}}
\newcommand\gmorph[3]{#1 : #2 \mkern.5mu \bar{\bigtriangleup} \mkern.4mu #3}
\begin{definition}
We introduce an ordering variant of the search tree logic:
\[
\gmorph{w}{(\pseq{},\bseq{},\sseq{})}{(\pseq{}',\bseq{}',\sseq{}')}
\]
with the following inference rules:
\[
\infer{\treemor{a}{\pseq{1}\bseq{1}}{\pseq{2}} \quad \treemor{a}{\sseq{1}}{\pseq{3}} \quad (\empseq,\pseq{3}) \pgrpx (\bseq{2},\sseq{2})}{\gmorph{a}{(\pseq{1},\bseq{1},\sseq{1})}{(\pseq{2},\bseq{2},\sseq{2})}}
\]
\[
\infer{\gmorph{w}{(\pseq{1},\bseq{1},\sseq{1})}{(\pseq{2},\bseq{2},\sseq{2})} \qquad \gmorph{a}{(\pseq{2},\bseq{2},\sseq{2})}{(\pseq{3},\bseq{3},\sseq{3})}}{\gmorph{wa}{(\pseq{1},\bseq{1},\sseq{1})}{(\pseq{3},\bseq{3},\sseq{3})}}
\]
\end{definition}

The $\gtriangle$ semantics recursively re-arranges the search tree into $\pseq{}$, $\bseq{}$ and $\sseq{}$ components at each depth. A derivation using the $\gtriangle$ semantics may be visualized as in Figure~\ref{fig_ex_gtriangle}. Note that in this hypothetical derivation, we encounter repeated states at depth $w_1$, thus giving rise to the first non-empty $\bseq{}$ component ($\bseq{1}$). From $w_1$ to $w_1w_2$, we have non-empty $\pseq{}$ and $\sseq{}$ components. Again at depth $w_1w_2$ we can observe a non-empty $\bseq{}$ component, which is the result of the previous $\sseq{}$ component generating duplicates at this depth. The $\pseq{}$ component can be thought of as the shadow/projection of all the previous $\bseq{}$ components.

\begin{figure}
\begin{center}
\begin{tikzpicture}[scale=0.7]
\coordinate (lbot) at (0,0);
\coordinate (rbot) at (10,0);
\coordinate (top) at (5,5);
\node[above] at (top) {$p$};
\draw (lbot) -- (top) -- (rbot) -- (lbot);
\coordinate (h1) at (3,3);
\coordinate (h2) at (7,3);
\coordinate (h3) at (1,1);
\coordinate (h4) at (9,1);
\coordinate (x1) at (2,0);
\coordinate (x2) at (4,0);
\coordinate (x3) at (6,0);
\coordinate (i2) at (intersection of h1--h2 and top--x2);
\coordinate (i3) at (intersection of h3--h4 and top--x2);
\coordinate (i4) at (intersection of h3--h4 and top--x3);
\draw (i2) -- (x2);
\draw (h1) -- (i2);
\node[above,scale=.8] at ($ (h1) !.5! (i2) $) {$\bseq{1}$};
\draw [dotted] (i2) -- (h2);
\node[above,scale=.8] at ($ (i2) !.5! (h2) $) {$\sseq{1}$};
\draw (i4) -- (x3);
\draw [dashed] (h3) -- (i3);
\node[above,scale=.8] at ($ (h3) !.5! (i3) $) {$\pseq{2}$};
\draw (i3) -- (i4);
\node[above,scale=.8] at ($ (i3) !.5! (i4) $) {$\bseq{2}$};
\draw [dotted] (i4) -- (h4);
\node[above,scale=.8] at ($ (i4) !.5! (h4) $) {$\sseq{2}$};
\draw[|-] (11,0) -- (11,1);
\node[right] at (11,.5) {$w_3$};
\draw[|-] (11,1) -- (11,3);
\node[right] at (11,2) {$w_2$};
\draw[|-|] (11,3) -- (11,5);
\node[right] at (11,4) {$w_1$};
\end{tikzpicture}
\caption{An example $\gtriangle$ derivation.}
\label{fig_ex_gtriangle}
\end{center}
\end{figure}

\begin{lemma}\label{lem_morph_relation}
Let $p$ be a state and $w$ an input string such that:
\[
\treemor{w}{p}{\pseq{}} \qquad 
\gmorph{w}{(\empseq, \empseq, p)}{(\pseq{}', \bseq{}, \sseq{})} 
\]
Then $\pseq{}'\bseq{}\sseq{} \psim \pseq{}$.
\end{lemma}
\begin{proof}
By induction on the length of $w$. For the base case ($w = a$), suppose $\treemor{a}{p}{\pseq{}}$. Then from the definition of $\gtriangle$ we get:
\[
\gmorph{a}{(\empseq, \empseq, p)}{(\empseq, \bseq{}, \sseq{})}
\]
Where $(\empseq, \pseq{}) \pgrpx (\bseq{}, \sseq{})$. Therefore, Lemma~\ref{lem_pgrpx_props} (a) gives $\pseq{} \psim \bseq{}\sseq{}$. For the inductive step ($w = w'a$), suppose:
\begin{align*}
&\treemor{w'}{p}{\pseq{1}} \tag{A.1}\\
&\gmorph{w'}{(\empseq, \empseq, p)}{(\pseq{1}', \bseq{1}, \sseq{1})} \tag{A.2}
\end{align*}
Then the induction hypothesis yields $\pseq{1} \psim \pseq{1}'\bseq{1}\sseq{1}$. Now let us assume:
\begin{align}
&\treemor{a}{\pseq{1}}{\pseq{2}} \tag{B.1}\\
&\treemor{a}{\pseq{1}'\bseq{1}}{\pseq{2}'}\tag{B.2}\\
&\treemor{a}{\sseq{1}}{\pseq{3'}}\tag{B.3}\\
&(\empseq, \pseq{3}') \pgrpx (\bseq{2}, \sseq{2})\tag{B.4}
\end{align}
Assumptions (A.1), (B.2) - (B.4) and the definition of $\gtriangle$ leads to:
\[
\gmorph{w'}{(\empseq, \empseq, p)}{(\pseq{2}', \bseq{2}, \sseq{2})}
\]
Moreover, assumptions (B.2), (B.3) implies
\[
\treemor{a}{\pseq{1}'\bseq{1}\sseq{1}}{\pseq{2'}\pseq{3'}}
\]
That is, we have:
\[
\pseq{1} \psim \pseq{1}'\bseq{1}\sseq{1} \tag{I.H}
\]
\[
\treemor{a}{\pseq{1}}{\pseq{2}} \tag{B.1}
\]
\[
\treemor{a}{\pseq{1}'\bseq{1}\sseq{1}}{\pseq{2'}\pseq{3'}}
\]
Applying Lemma~\ref{lem_triangle_psim} to these three relations yield $\pseq{2} \psim \pseq{2'}\pseq{3'}$. Furthermore, Lemma~\ref{lem_pgrpx_props} (a) implies (with B.4) $\pseq{3'} \psim \bseq{2}\sseq{2}$. Finally, Lemma~\ref{lem_psim_props} (c) gives $\pseq{2} \psim \pseq{2'}\bseq{2}\sseq{2}$ as required.
\end{proof}

\begin{lemma}\label{lem_gtree_props}
Let $\gamma$ be a path with $p = \gcod{\gamma}$ and $w$ an input string such that:
\[
\gmorph{w}{(\empseq, \empseq, p)}{(\pseq{}, \bseq{}, \sseq{})}
\]
Then the following properties hold:
\[
\elems{\pseq{}\bseq{}\sseq{}} \subseteq \gfree{\gamma} \tag{a}
\]
\[
\length{\sseq{}} \le \card{\gfree{\gamma}} \tag{b}
\]
\[
\length{\bseq{}\sseq{}} \le \card{\gfree{\gamma}} * o \tag{c}
\]
Where $o$ is the fan-out of the NFA.
\end{lemma}
\begin{proof}
For property (a), suppose $\treemor{w}{p}{\pseq{}'}$. From Lemma~\ref{lem_sib_restrict} we get $\elems{\pseq{}'} \subseteq \gfree{\gamma}$. Moreover, Lemma~\ref{lem_morph_relation} gives $\pseq{}\bseq{}\sseq{} \psim \pseq{}'$. Now, Lemma~\ref{lem_psim_props} (a) gives $\pseq{}\bseq{}\sseq{} \subseteq \gfree{\gamma}$.

For property (b), note that it follows from property (a) that $\elems{\sseq{}} \subseteq \gfree{\gamma}$. From the definition of $\gtriangle$ it follows that $\exists \; \pseq{}' \;.\; (\empseq, \pseq{}') \pgrpx (\bseq{}, \sseq{})$. Therefore, from Lemma~\ref{lem_pgrpx_props} (d) we get $\length{\sseq{}} = \card{\elems{\sseq{}}} \le \card{\gfree{\gamma}}$.

For property (c), suppose $w = w'a$ (the result holds trivially for $w = \empseq$). Then from the definition of $\gtriangle$ there exist $\sseq{}', \pseq{}'$ such that:
\begin{align}
&\gmorph{w'}{(\empseq, \empseq, p)}{(\_, \_,\sseq{}')} \tag{A.1}\\
&\treemor{a}{\sseq{}'}{\pseq{}'} \tag{A.2}\\
&(\empseq, \pseq{}') \pgrpx (\bseq{}, \sseq{}) \tag{A.3}
\end{align}
From (A.2) and the structure of the NFA, we derive $\length{\pseq{}'} \le \length{\sseq{}'} * o$. Furthermore, (A.3) and Lemma~\ref{lem_pgrpx_props} (a) implies $\bseq{}\sseq{} \psim \pseq{}'$. Therefore, Lemma~\ref{lem_psim_props} (a) and property (b) above leads to $\length{\bseq{}\sseq{}} = \length{\pseq{}'} \le \length{\sseq{}'} * o \le \card{\gfree{\gamma}} * o$.
\end{proof}

\begin{lemma}\label{lem_gtriangle_fringe}
Let $w$ be an input string and p a state. Let $k$ be a constant offset into $w$ and $i$ an index such that:
\[
0 < i \le k \le \length{w}
\]
\[
\gmorph{w[0:i]}{(\empseq, \empseq, p)}{(\pseq{i}, \bseq{i}, \sseq{i})}
\]
\[
\treemor{w[i:k]}{\bseq{i}}{\bseq{(i, k)}}
\]
Then $\pseq{k} = \bseq{(1,k)} \ldots \bseq{(k - 1, k)}$
\end{lemma}
\begin{proof}
By induction on $k$ (omitted).
\end{proof}

With reference to Figure~\ref{fig_ex_gtriangle}, Lemma~\ref{lem_gtriangle_fringe} establishes the connection between the fringe of the overall triangle and those of individual trapezoidal slices (the concatenation of the bases of the trapezoids make up the base of the overall triangle).

\begin{lemma}\label{lem_slice_props}
Let $\gamma$ be a path with $p = \gcod{\gamma}$ and $w$ an input string such that:
\[
\gmorph{w}{(\empseq, \empseq, p)}{(\pseq{}, \bseq{}, \sseq{})}
\]
Then for any $q \pin \bseq{}$, there exists a path $\gamma'$ from $p$ to $q$ such that $\gfree{\gamma\gamma'} \subsetneq \gfree{\gamma}$.
\end{lemma}
\begin{proof}
Follows from Lemma~\ref{lem_sib_restrict}; $q$ is repeated within $\bseq{}$, making $p \in \gshared{\gamma\gamma'}$.
\end{proof}
%

\begin{lemma}\label{lem_poly_bound}
Suppose $\gamma$ is a path with $p = \gcod{\gamma}$ and $w$ an input string of length $n$ such that:
\[
\treemor{w}{p}{\pseq{}}
\]
Then the following holds:
\[
\length{\pseq{}} < k^k * o^k * n^k
\]
Where $k = \card{\gfree{\gamma}}$.
\end{lemma}
\begin{proof}
By induction on $k$.

\vspace{1em}
\noindent
\textbf{Base case - 1}: Suppose $k = 0$. Then it follows from Lemma~\ref{lem_sib_restrict} that $\length{\pseq{}} = 0$, which is within the bounds of our polynomial.

\vspace{1em}
\noindent
\textbf{Base case - 2}: Suppose $k = c$ (for some constant $c$) and:
\[
\not\exists \gamma' \;.\; \gamma' = \gamma\gamma'' \land \gfree{\gamma'} < c
\]
This means the search tree rooted at $\gcod{\gamma}$ cannot contain duplicates at any depth, for if it does, we can always find an extended path $\gamma'$ for which $\gfree{\gamma'}$ is less. This restriction immediately implies that the fringe of the search tree cannot grow beyond $c$, which is well within the bounds of our (over-estimating) polynomial ($c^c * o^c * n^c$).

\vspace{1em}
\noindent
\textbf{Inductive step}: Suppose:
\[
\gmorph{w[0:i]}{(\empseq, \empseq, p)}{(\pseq{i}, \bseq{i}, \sseq{i})}
\]
\[
\treemor{w[i:n]}{\bseq{i}}{\bseq{(i,n)}}
\]
Where $0 < i \le n$. From Lemma~\ref{lem_gtriangle_fringe} we deduce:
\begin{align*}
\pseq{n}\bseq{n}\sseq{n} = \bseq{(1,n)}\ldots\bseq{(n-1, n)}\bseq{n}\sseq{n} \tag{A}
\end{align*}
It follows from Lemma~\ref{lem_slice_props} that we can apply the induction hypothesis to each path ending in some state within an $\bseq{i}$. Therefore, we derive:
\[
\forall i \; \exists v < k \;.\; \length{\bseq{(i, n)}} < \length{\bseq{i}} * {v}^{v} * o^{v} * \length{w[i:n]}^{v}
\]
In terms of the illustration in Figure~\ref{fig_ex_gtriangle}, this statement measures the bottom edges of the trapezoids. Now, taking into account that $v < k$ and $\length{w[i:k]} < n$, we arrive at:
\[
\forall i \;.\; \length{\bseq{(i, n)}} < \length{\bseq{i}} * {k}^{k} * o^{k} * n^{k}
\]
Moreover, it follows from Lemma~\ref{lem_gtree_props} (c) that $\length{\bseq{i}} \le k * o$. Therefore, we get:
\begin{align*}
\forall i \;.\; \length{\bseq{(i, n)}} < {k}^{k + 1} * o^{k + 1} * n^{k} \tag{B}
\end{align*}
Now, we combine (A) and (B) to obtain:
\[
\length{\pseq{n}\bseq{n}\sseq{n}} < (n - 1) * {k}^{k + 1} * o^{k + 1} * n^{k} + \length{\bseq{n}\sseq{n}}
\]
Furthermore, it follows from Lemma~\ref{lem_gtree_props} (c) that:
\[
\length{\bseq{n}\sseq{n}} \le k * o < k^{k + 1} * o^{k + 1} * n^k
\]
Therefore, we get:
\[
\length{\pseq{n}\bseq{n}\sseq{n}} < {k}^{k + 1} * o^{k + 1} * n^{k + 1}
\]
Since we know $\pseq{} \psim \pseq{n}\bseq{n}\sseq{n}$ from Lemma~\ref{lem_morph_relation}, the inductive step holds.
\end{proof}

\begin{theorem}[Redos analysis completeness]\label{thm_complete}
Given an NFA with an exponential runtime vulnerability, the REDoS analysis presented in Section~\ref{analysis} will produce an attack string which triggers this behaviour on a backtracking regular expression matcher.
\end{theorem}
\begin{proof}
Lemma~\ref{lem_poly_bound} implies that for a non-pumpable NFA, the search tree width is polynomially bounded. Since $w$ is finite, the entire search space in turn becomes polynomially bounded. This suggests that only a pumpable NFA can lead to an exponentially large search space.

The analysis presented in Section~\ref{analysis} derives an attack string of the form:
\[
\begin{array}{rcll}
x &:& \startnode \ordtrans (\pseq{}\;\loopnode\;\pseq{}')\\
y &:& \elems{\pseq{}\,\loopnode} \dfaarrow \mstat{y}
& \qquad \mstat{y} \subseteq \elems{\pseq{}\,\loopnode}\\
z &:& \mstat{y} \dfaarrow {\failstate}
& \qquad \failstate \cap\Acc=\emptyset
\end{array}
\]
For completeness, we must show that all of these conditions (in addition to a pumpable node $\loopnode$) are necessary for a backtracking matcher to have an exponential runtime.

\vspace{1em}
\noindent
\textbf{Case - 1}: Suppose the analysis cannot find a prefix $x$ such that $x : \startnode \ordtrans \pseq{}\;\loopnode\;\pseq{}'$. Given that the prefix analysis is exhaustive (in that, it iterates through all the unique ordered multi-states corresponding to $Q$), it would only fail to find a prefix for some $\loopnode$ only if that node is un-reachable in the NFA. This in turn implies that a backtracking machine would not have to explore the exponentially large search tree emanatig from $\loopnode$.

\vspace{1em}
\noindent
\textbf{Case - 2}: Now suppose the analysis finds a prefix $x$ of the form:
\[
x : \startnode \ordtrans (\pseq{}\;\loopnode\;\pseq{}')
\]
The next step of the analysis is to find a pumpable string $y$ of the form:
\[
y : \elems{\pseq{}\,\loopnode} \dfaarrow \mstat{y} \qquad \mstat{y} \subseteq \elems{\pseq{}\,\loopnode}\\
\]
Here we have to establish that a pumpable string of this form must be available for any vulnerable NFA.

It follows from Lemma~\ref{lem_poly_bound} that any successful attack string would have to exploit a pumpable node in order to generate an exponentially large search space. Suppose $x\rho_i$ is such an attack string, so that:
\[
\begin{array}{rcll}
x &:& \startnode \ordtrans (\pseq{}\;\loopnode\;\pseq{}')\\
\rho_i &=& a_0a_1a_2 \ldots a_i \qquad a_i \in Z \\
\rho_i &:& \elems{\pseq{}\,\loopnode} \dfaarrow \mstat{\rho_i}
\end{array}
\]
In other words, the attack string $x\rho_i$ must allow the $\rho_i$ component to be extended while maintaining the exponential search-space expansion. Given the finiteness of $Q$, it immediately follows that:
\[
\exists j, k \; . \; (j < k) \land \mstat{\rho_k} \subseteq \mstat{\rho_j}
\]
That is, as we keep extending $\rho_i$, there must come a point where $\mstat{\rho_i}$ shrinks.

We now have:
\[
\begin{array}{rcll}
x &:& \startnode \ordtrans (\pseq{}\;\loopnode\;\pseq{}')\\
\rho_j &:& \elems{\pseq{}\,\loopnode} \dfaarrow \mstat{\rho_j} \\
a_{j + 1} \ldots a_k &:& \mstat{\rho_j} \dfaarrow \mstat{\rho_k} \qquad \mstat{\rho_k} \subseteq \mstat{\rho_j}
\end{array}
\]
Now suppose $x\rho_j : \startnode \ordtrans (\bar{\pseq{}}\;\loopnode\;\bar{\pseq{}'})$. It follows from Lemma~\ref{dfatree} and Lemma~\ref{ordtranstree} that $\elems{\bar{\pseq{}}\;\loopnode} = \mstat{\rho_j}$. Finally, allowing $y = a_{j + 1} \ldots a_k$, we arrive at:
\[
\begin{array}{rcll}
x\rho_i &:& \startnode \ordtrans (\bar{\pseq{}}\;\loopnode\;\bar{\pseq{}'})\\
y &:& \elems{\bar{\pseq{}}\;\loopnode} \dfaarrow \mstat{\rho_k} \qquad \mstat{\rho_k} \subseteq \elems{\bar{\pseq{}}\,\loopnode}
\end{array}
\]
That is, for a valid attack string to exist, there must also exist a (\textit{prefix, pumpable string}) pair that meets the criteria of the analysis. The fact that the analysis rejects some configurations not meeting the stability condition is largely irrelevant; the analyser is bound to find a suitable (\textit{prefix, pumpable string}) pair eventually (if a valid attack string configuration does indeed exist).

\vspace{1em}
\noindent
\textbf{Case - 3}: Let $x, y$ be constructed from the analysis such that:
\[
\begin{array}{rcll}
x &:& \startnode \ordtrans (\pseq{}\;\loopnode\;\pseq{}')\\
y &:& \elems{\pseq{}\,\loopnode} \dfaarrow \mstat{y}
& \qquad \mstat{y} \subseteq \elems{\pseq{}\,\loopnode}
\end{array}
\]
The final step of the analysis is to construct a failure suffix $z$ such that:
\[
z : \mstat{y} \dfaarrow {\failstate} \qquad \failstate \cap\Acc=\emptyset
\]
It is immediately clear that such a suffix must be available for a valid attack string to exist, as otherwise a backtracking machine would lead to acceptance and not have to explore the exponentially large search space constructed by the $(x, y)$ pair.
\end{proof}

\lstset{ %
  basicstyle=\tt \footnotesize,        
  breaklines=true,                 
  commentstyle=\color{green},    
  keepspaces=true,                 
  showspaces=false,                
  showstringspaces=false,          
  showtabs=false,                  
  stepnumber=2,                    
  tabsize=2,                       
  title=\lstname,                   
  frame=single
}
\section{Implementation}\label{implementation}
We implemented the analysis presented above in OCaml~\cite{rxxr2} (code-named \textit{RXXR}). Apart from the code used for parsing regular expressions (and some other boilerplate code), the main source modules have an almost one-to-one correspondence with the concepts discussed thus far. This relationship is illustrated in Table~\ref{fig_impl_correspondence}.

\begin{figure}[ht]
\begin{center}
\begin{tabular}{ll}
\hline
Concept (Theory)& Implementation (OCaml Module)\\
\hline
\hline
NFA& \texttt{Nfa.mli/ml}\\
$\pseq{}$& \texttt{Beta.mli/ml}\\
$\mstat{}$& \texttt{Phi.mli/ml}\\
$\twoarrow$& \texttt{Product.mli/ml}\\
$\threearrow$& \texttt{Triple.mli/ml}\\
Prefix analysis& \texttt{XAnalyser.mli/ml}\\
Pumpable analysis ($y_1$)& \texttt{Y1Analyser.mli/ml}\\
Pumpable analysis ($ay_2$)\;\;& \texttt{Y2Analyser.mli/ml}\\
Suffix analysis& \texttt{ZAnalyser.mli/ml}\\
Overall analysis& \texttt{AnalyserMain.mli/ml}\\
\hline
\hline
\end{tabular}
\end{center}
\caption{Theory to source-code correspondence}
\label{fig_impl_correspondence}
\end{figure}

Each module interface (\texttt{.mli} file) contains function definitions which directly correspond to various aspects of the analysis presented earlier. For an example, the NFA module provides the following function for querying ordered transitions:
\vspace{-1.5em}
\begin{lstlisting}
val get_transitions : Nfa.t -> int ->
                    ((char * char) * int) list;;
\end{lstlisting}
The NFA states are represented as integers. Each symbol of the input alphabet is encoded as a pair of characters, allowing a uniform representation of character classes (\texttt{[a-z]}) as well as individual characters.

The NFA used in the implementation (\texttt{Nfa.mli/ml}) contains $\varepsilon$ transitions, which were not part of the NFA formalization presented earlier. The reason for this deviation is that having $\varepsilon$ transitions allows us to preserve the structure of the regular expression within the NFA representation, which in turn preserves the order of the transitions. The correctness of the implementation is unaffected as the two forms of NFA representation are semantically equivalent (i.e. represents the same language). Only a slight mental adjustment (from ordered NFAs to $\varepsilon$-NFAs) is required to correlate the theoretical formalizations to the OCaml code. For an example, Figure~\ref{fig_beta_mli} presents the module interface for $\pseq{}$.
\small
\begin{figure}
\begin{lstlisting} 
(* internal representation of beta *)
type t;;

module BetaSet : (Set.S with type elt = t);;

(* beta with just one state *)
val make : int -> t;;

(* returns the set of states contained within this beta *)
val elems : t -> IntSet.t;;

(* calculate all one-character reachable betas *)
val advance : (Nfa.t * Word.t * t) -> (Word.t * t) list;;

(* consume all epsilon transitions while recording pumpable kleene encounters *)
val evolve : (Nfa.t * Word.t * t) -> IntSet.t -> 
               Flags.t * t * (int * t) list;;
\end{lstlisting}
\caption{\texttt{Beta.mli}}
\label{fig_beta_mli}
\end{figure}
\normalsize
The function \texttt{advance()} is utilized inside the \texttt{XAnalyser.ml} module to perfom the closure computation (i.e. compute all $\pseq{}$s reachable from the root node), whereas \texttt{evolve()} is a utility function used to work around the $\varepsilon$ transitions. The modules \texttt{(Phi / Product / Triple).mli} define similar interfaces for $\Phi, \twoarrow$ and $\threearrow$ constructs introduced in the analysis. 

The different phases of the analysis is implemented inside the corresponding analyser modules. As an example, Figure~\ref{fig_y2_mli} presents the \texttt{Y2Analyser.mli} module responsible for carrying out the analysis after the branch point ($\threearrow$ simulation).
\small
\begin{figure}
\begin{lstlisting} 
(* internal representation of the analyser *)
type t;;

(* initialize analyser instance for the specified triple and the kleene state *)
val init : (Nfa.t * Word.t * Triple.t) -> int -> t;;

(* calculate the next (y2, phi) *)
val next : t -> (Word.t * Phi.t) option;;

(* read analyser flags *)
val flags : t -> Flags.t;;
\end{lstlisting}
\caption{\texttt{Y2Analyser.mli}}
\label{fig_y2_mli}
\end{figure}
\normalsize
The internal representation of the analyser (\texttt{type t}) holds the state of the closure computation, which is initialized with an initial triple argument through the \texttt{init()} function. We defer the interested reader to module definition (\texttt{.ml}) files for further details on the implementation. 

\subsection{Evaluation data}\label{rxxr:eval_data}
The analysis was tested on two corpora of regexes. The first of these was extracted from an online regex library called \textit{RegExLib}~\cite{2012_regexlib}, which is a community-maintained regex archive; programmers from various disciplines submit their solutions to various pattern matching tasks, so that other developers can reuse these expressions for their own pattern matching needs. The second corpus was extracted from the popular intrusion detection and prevention system \textit{Snort}~\cite{2012_snort}, which contains regex-based pattern matching rules for inspecting IP packets across network boundaries. The contrasting purposes of these two corpora (one used for casual pattern matching tasks and the other used in a security critical application) allow us to get a better view of the seriousness of exponential vulnerabilities in practical regular expressions.

The regex archive for RegExLib was only available through the corresponding website~\cite{2012_regexlib}. Therefore, as the first step the expressions had to be scraped from their web source and adapted so that they can be fed into our tool. These adaptations include removing unnecessary white-space, comments and spurious line breaks. A detailed description of these adjustments as well as copies of both adjusted and un-adjusted data sets have been included with the resources linked from the RXXR distribution~\cite{rxxr2} (also including the Python script used for scraping). The regexes for Snort, on the other hand, are embedded within plain text files that define the Snort rule set. A Python script (also linked from the RXXR webpage) allowed the extraction of these regexes, and no further processing was necessary.

\subsection{Results}
\label{results}
The results of running the analysis on these two corpora of regexes are presented in Table~\ref{fig_analysis_results}. The figures  show that we can process  around 75\% of each of the corpora with the current level of syntax support. Out of these analyzable amounts, it is notable that regular expressions from the RegExLib archive use the Kleene operator more frequently (about 50\% of the analyzable expressions) than those from the Snort rule set (close to 30\%). About 11.5\% of the Kleene-based RegExLib expressions were found to have a pumpable Kleene expression as well as a suitable suffix, whereas for Snort this figure stands around 0.55\%.

\begin{figure}[h]
\begin{center}
\begin{tabular}{lrr}
\hline
& RegExLib & Snort\\
\hline
\hline
Total patterns &2992 &12499\\
Parsable &2290 &9801\\
Pumpable &159 &19\\
Vulnerable &131 &15\\
Interrupted &4 &0\\
Pruned &0 &2\\
Time &61.51 (s) &30.10 (s)\\
\hline
\end{tabular}
\end{center}
\caption{RXXR2 results - statistics}
\label{fig_analysis_results}
\end{figure}

The tool makes every attempt to analyse a given pattern, even the ones which contain non-regular constructs like backreferences. An expression $(e_1|e_2)$ may be vulnerable due to a pumpable Kleene that occurs within $e_1$, whereas $e_2$ might contain a backreference. In these situations, the analyser attempts to derive an attack string which avoids the non-regular construct. If such a non-regular construct cannot be avoided, the analysis is terminated with the \texttt{interrupted} flag.

\phantomsection\label{rxxr:unstable_example}
On certain rare occasions, search pruning is employed as an optimization. It is activated when there have been a number of unstable derivations (failing to meet $\mstat{y2} \subseteq \mstat{x}$) for a given prefix. For an example, consider the regular expression:
\[
\kleene{(\kleene{[\text{\^{}}a]}b)}[\text{\^{}}c]\{1000\}
\]
Here the Kleene expression $\kleene{(\kleene{[\text{\^{}}a]}b)}$ is pumpable for any string which contains two copies of $b$ (e.g. $bb, bab, abb, cbb \ldots$). However, if the analysis were to pick a pumpable string that \emph{does not contain the symbol $c$}, it will lead to an unstable derivation. Intuitively, the followup expression $[\text{\^{}}c]\{1000\}$ (which has a large state space) will also consume the pumpable string and introduce a new state in $\mstat{y2}$, breaking the inclusion $\mstat{y2} \subseteq \mstat{x}$. Pruning allows the analysis to attempt different variants of the pumpable string without getting stuck on a single search path where all of the pumpable strings lead to unstable (but unique) derivations (e.g. $bb, bab, baab, baaab, \dots$). Needless to say, this is an ad-hoc optimization that can be further improved with more sophisticated heuristics. Given that pruning was only triggered in two instances for the entire data set above, we believe the current heuristic (a static bound on the number of unstable derivations) is adequate. If a pruned search does not report a vulnerability, it should be re-run with a higher (or infinite) prune limit in order to obtain a conclusive result.

\subsubsection{Validation}\label{rxxr_validation}
The task of validating vulnerabilities is complicated by the fact that different regular expression implementations (Java, Python, .NET etc.) have different syntax flavours. RXXR itself is written to accept PCRE like patterns of the form \texttt{/<REGEX>/<FLAGS>} where \texttt{REGEX} contains the main expression and \texttt{FLAGS} are used to control various aspects of the matching process (e.g. whether to match multi-line input or not). Java, Python and .NET use separate library calls to configure such behavior. Moreover, they can also differ from one another in terms of the syntax allowed within the main expression. For an example, Java requires tricky escape sequences when working with meta-characters (e.g a literal backslash requires \verb|\\\\|), whereas Python is more flexible with its support for raw (un-interpreted) input strings.

For these reasons we chose Python as our main validation platform (Python's support for raw strings makes the porting relatively simple). A sample of vulnerabilities were then manually validated on other platforms (Java, .NET and PCRE). Table~\ref{fig_validation_python} illustrates how Python responds to above vulnerabilities.

\begin{figure}[ht]
\begin{center}
\begin{tabular}{lrr}
\hline
& RegExLib & Snort\\
\hline
\hline
Total vulnerabilities &131 &15\\
Successfully validated &115 &14\\
Python parsing bug &12 &0\\
Python not vulnerable &4 &1\\
\hline
\end{tabular}
\end{center}
\caption{Validation of vulnerabilities - Python}
\label{fig_validation_python}
\end{figure}

The Python scripts developed for this validation are also included with the RXXR distribution~\cite{rxxr2}, along with instructions on how to reproduce the above results. We discovered that Python was not able to compile regular expressions of the form $\kleene{(\kleene{[a-z]})}$, which is a known Python defect~\cite{python_parse_bug}. Variants of this bug affected 12 of the RegExLib vulnerabilities which we could not validate on Python. The remaining few cases were down to trivial vulnerabilities that Python manages to work around. We observed that both Python and .NET are capable of avoiding vulnerabilities in expressions like $\kleene{([a-c]|b)}d$ or $\kleene{(a|a)}b$, where the redundancies are quite obvious. Interestingly however, Java does not seem to implement any such workarounds; even when matching the expression $\kleene{(a|a)}b$ against the input string $a^n (n \sim 50)$, the JVM (Java Virtual Machine) becomes non-responsive.

\subsubsection{Sample vulnerabilities}
The vulnerabilities reported range from trivial programming errors to more complicated cases. For an example, the following regular expression is meant to validate time values in 24-hour format (from RegExLib):
\begin{verbatim}
    ^(([01][0-9]|[012][0-3]):([0-5][0-9]))*$
\end{verbatim}
Here the author has mistakenly used the Kleene operator instead of the \texttt{?} operator to suggest the presence or non-presence of the value. This pattern works perfectly for all intended inputs. However, our analysis reports that this expression is vulnerable with the pumpable string ``\verb|13:59|" and the suffix ``\verb|/|". This result gives the programmer a warning that the regular expression presents a DoS security risk if exposed to user-malleable input strings to match.

For a moderately complicated example,  consider the following regular expression (again from RegExLib):
\small
\begin{verbatim}
^([a-zA-z]:((\\([-*\.*\w+\s+\d+]+)|(\w+)\\)+)(\w+.zip)|(\w+.ZIP))$
\end{verbatim}
\normalsize
This expression is meant to validate file paths to zip archives. Our tool identifies this expression as vulnerable and generates the prefix ``\verb|z:\ |", the pumpable string ``\verb|\zzz\|" and the empty string as the suffix. This is probably an unexpected input in the author's eye, and this is another way in which our tool can be useful in that it can point out potential mis-interpretations which may have materialized as vulnerabilities.

Out of the over 12,000 patterns examined, there were two cases that failed to terminate within any reasonable amount of time. Closer inspection reveals that a pumpable Kleene expression with a vast number of states is to blame. Consider the following example (from RegExLib):
\small
\begin{verbatim}
^(([a-zA-Z0-9_\-\.]+)@([a-zA-Z0-9_\-\.]+)\.
  ([a-zA-Z]{2,5}){1,25})+
  ([;.](([a-zA-Z0-9_\-\.]+)@([a-zA-Z0-9_\-\.]+)\.
  ([a-zA-Z]{2,5}){1,25})+)*$
\end{verbatim}
\normalsize
If we change the counted expressions of the form \verb|e{1,25}| into \verb|e{1,5}|, the analyser returns immediately. This shows that the analysis itself can take a long time on certain inputs. However, such cases are extremely rare.

\subsection{Comparison to fuzzers}
REDoS analysers commonly used in practice are based on a brute-force approach known as fuzzing,  where the runtime of a pattern is tested against a set of strings. A leading example of this approach is the Microsoft's SDL Regex Fuzzer~\cite{2011_regex_fuzzer}.

As is common with most brute-force approaches, the main problem with fuzzing is that it can take a considerable amount of time to detect a vulnerability. This is especially pronounced in the case of REDoS analysis as vulnerable patterns tend to take increasing amounts of time with each iteration of testing. This property alone disqualifies fuzzing based REDoS analysers from being integrated into code-analysis tools, as their operation would impose unacceptable delays. For an example, consider the following simple pattern:
\begin{verbatim}
                   ^(a|b|ab)*c$
\end{verbatim}
Even with a lenient fuzzer configuration (ASCII only, 100 fuzzing iterations), SDL fuzzer takes 5-10 minutes to report a vulnerability on this pattern. By comparison, our analyser can process tens of thousands of patterns in less time.

Fuzzers can also miss out on vulnerabilities. For an example, consider the following two patterns:
\begin{verbatim}
                  ^.*|(a|b|ab)*c$
                  ^(a|b|ab)*c|.*$
\end{verbatim}
SDL Fuzzer reports both of these patterns as being safe. However, the non-commutative property of the alternation renders the second pattern vulnerable (as explained in Section~\ref{examples}).  Another such example is:
\begin{verbatim}
                ^(a|b|c|ab|bc)*a.*$ 
\end{verbatim}
For this pattern, only one of the pumpable strings ($bc$) can lead to an attack string, and it must not end in an $a$. Such relationships are difficult to be caught in a heuristics-based fuzzer.

Yet another problem with fuzzers is caused by the element of randomness present in their string generating algorithms. Since fuzzers are not based on any sound theory, some form of randomness is necessary in order to increase the chance of stumbling upon a valid attack string. However, this can make the fuzzer yield inconsistent results for the same pattern. Consider the following pattern for an example:
\begin{verbatim}
            (a|b)*[^c].*|(c)*(a|b|ab)*d
\end{verbatim}
The SDL fuzzer reports this pattern as being safe in most invocations, but in few cases it finds an attack string. 

Finally, the ultimate purpose of using a static analyser is to detect potential vulnerabilities upfront and lead to the corresponding fixes. Our analyser pin-points the exact pumpable Kleene expression and generates a string (pumpable string) which witnesses vulnerability, making the fixing of the error a straightforward task. This is notably in contrast to the fuzzer, which outputs a random string (mostly in hex format) that does not provide any insight into the source of the problem.

\section{Related work}\label{related}
The starting point for the present paper was the regular expression analysis RXXR \cite{analysisredos}. While that paper was aimed at a security audience, the present paper complements it by using a programming language approach inspired by type theory and logic.

Program analysis for security is by now a well established field~\cite{chess2004static}. REDoS is known in the literature as a special case of algorithmic complexity attacks~\cite{2003_dos_crosby,2006_btrack_smith}. Sugiyama and Minamide~\cite{satoshi_2014} has recently explored a static analysis for checking wheather a given regular expression matching operation can be performed in linear time. Their analysis assumes that the input string is available at analysis time, which could make it less efficient in that the analysis must be run each time prior to the actual matching operation. Nevertheless, we can imagine such an analysis acting as a white-list filter for input strings that are accepted by a system. On the other hand, our analysis can be run off-line, allowing system administrators to filter out vulnerable regular expressions much earlier.
 
Parsing Expression Grammars (PEGs) have been proposed as an alternative to regular expressions~\cite{fordpeg} that avoid their nondeterminism. In a series of tutorials~\cite{2007_regex_cox,2009_regex_cox}, Cox has argued for Thompson's lockstep matcher~\cite{1968_thompson} as a superior alternative to backtracking. Cox's implementation of the lockstep algorithm, called RE2~\cite{2007_regex_cox}, guarantees linear time matching of purely regular expressions. While RE2 also supports few extensions to basic regular expressions (such as submatch extraction), backtracking matchers still have an advantage when it comes to complex irregular constructs, backreferences~\cite{1990_aho} being the case in point. We believe that this is one of the reasons why the backtracking implementations are is still commonplace. In any case, as long as backtracking matchers remain widely deployed, the REDoS problem will also remain with us for the foreseeable future.

Backtracking is a classic application of continuations, and regular expression matchers similar to the \machname{} have been investigated in the functional programming literature~\cite{2001_danvy_defunc,1999_harper_debug,frischcardelli}. Other recent work on regular expressions in the programming language community includes regular expression inclusion~\cite{hengleinregexp}  and submatching~\cite{sulzmannsubmatch}.
  
Apart from some basic constructions like the power DFA covered in standard textbooks~\cite{hopcroftullman}, we have not explicitly relied on automata theory. Instead, we regarded the matcher as an abstract machine that can be analyzed with tools from programming language research. Specifically, the techniques in this paper are inspired by substructural logics, such as Linear Logic~\cite{girardlinear,prooftandypes} and Separation Logic~\cite{ishtiaqohearn,reynoldslicssep}. Concerning the latter, it may be instructive to compare the sharing of $w$ or absence of sharing of $\pseq{}$ in \Cref{treemor} to the connective of Separation logic. 
In a conjunction, the heap $h$ is shared:
 \[
 \infer{h \models P_{1} \qquad h \models P_{2}
 }{
 h \models P_{1} \wedge P_{2}
 }
\]
By contrast, in a separating conjunction, the heap is split into disjoint parts that are not shared:
\[
 \infer{
 h_{1} \models P_{1}
 \qquad 
  h_{2} \models P_{2}
  \qquad
  h_{1} \cap h_{2} = \emptyset
 }{
h_{1} \cup h_{2} \models P_{1} \mathbin * P_{2}
 }
 \]
 
Tree-shaped data structures have been one of the leading examples of separation logic and variations of it, such as  Context Logic~\cite{contextlogic}.
However, a difference to the search trees we have used in this paper is that the whole search tree is not actually constructed as a data structure in memory. Rather, only a diagonal cut across it is maintained at any time in the \machname{}. The whole tree does not exist in memory, but only in space \emph{and} time, so to speak. In that regard the search trees are like parse trees, which the parser only needs to construct in principle by traversing them, and not necessarily as a data structure in memory complete with details of all nodes~\cite{appelbook,2007_dragonbook}.

Even though the \machname{} is sequential, parts of the analysis are reminiscent of transition systems in  process algebras, particularly running two or more automata in parallel (\Cref{productNP,productNNP}). Seen that way, the may and must part of the analysis are analogous to the two modalities $\langle a\rangle$ and $[a]$ in Henessy-Milner logic~\cite{henessymilnerlogic}.
 
\section{Directions for further research} \label{further}
At present, the analysis constructs attack strings when there is the possibility of exponential runtime. It should be possible to extend the analysis to compute a polynomial as an upper bound for the runtime when there is no REDoS vulnerability causing exponential runtime.

The efficiency of the analyser compares favorably with that of the Microsoft SDL Regex Fuzzer~\cite{2011_regex_fuzzer}.
Given that we are computing sets of sets of states, the analysis may explore a large search space. One may take some comfort from the fact that type checking and inference for functional programming languages can have high complexity in the worst case~\cite{mairson1989deciding,seidl1994haskell} that may not manifest itself in practice. Nonetheless, we aim to revisit the design of the analysis and optimize it.

Pruning the search space may lead to improvements in efficiency. An intriguing possibility is to implement the analysis on many-core graphics hardware (GPUs). Using the right data structure representation for transitions, GPUs can efficiently explore nondeterministic transitions in parallel, as demonstrated in the iNFAnt regular expression matcher~\cite{2010_infant}.

The search tree logic (\Cref{treemor}) may have independent interest and possible connections to other substructural logics such as Linear Logic~\cite{girardlinear,prooftandypes}, Separation Logic~\cite{ishtiaqohearn,reynoldslicssep}, Lambek's syntactic calculus~\cite{lambekcalculus}, or substructural calculi for parse trees~\cite{semparsing}. Search trees are dual to parse trees in the sense that the nodes represent a disjunction rather than a conjunction.

\bibliographystyle{elsarticle-num}

\end{document}